\documentclass[aps,pra,groupedaddress,superscriptaddress,nofootinbib,twocolumn,notitlepage,bibnotes]{revtex4-1}
\pdfoutput=1

\usepackage{newtxtext,newtxmath}

\let\coloneqq\relax

\usepackage[latin1]{inputenc}
\usepackage{amsthm}
\usepackage{amssymb}
\usepackage{amsmath}
\usepackage{bbold}
\usepackage{bbm}
\usepackage[dvipsnames,svgnames]{xcolor}
\usepackage[colorlinks=true, allcolors=MidnightBlue!70!black!70!TealBlue]{hyperref}
\usepackage{braket}
\usepackage{dsfont}
\usepackage{mathdots}
\usepackage{mathtools}
\usepackage{enumerate}
\usepackage[shortlabels]{enumitem}
\usepackage{csquotes}
\usepackage{stmaryrd}
\usepackage[cal=boondox]{mathalfa}
\usepackage{graphicx}
\usepackage{stackengine}
\usepackage{scalerel}
\usepackage{tensor}       
\usepackage{array}
\usepackage{makecell}
\newcolumntype{x}[1]{>{\centering\arraybackslash}p{#1}}
\usepackage{tikz}
\usepackage{pgfplots}
\usetikzlibrary{shapes.geometric, shapes.misc, positioning, arrows, arrows.meta, decorations.pathreplacing, decorations.pathmorphing, patterns, angles, quotes, calc}
\usepackage{booktabs}
\usepackage{xfrac}
\usepackage{siunitx}
\usepackage{centernot}
\usepackage{comment}
\usepackage{chngcntr}
\usepackage{caption}
\usepackage{subcaption}

\newtheorem{thm}{Theorem}
\newtheorem*{thm*}{Theorem}
\newtheorem{prop}[thm]{Proposition}
\newtheorem*{prop*}{Proposition}
\newtheorem{lemma}[thm]{Lemma}
\newtheorem*{lemma*}{Lemma}
\newtheorem{cor}[thm]{Corollary}
\newtheorem*{cor*}{Corollary}

\newtheorem*{cj*}{Conjecture}
\newtheorem{Def}[thm]{Definition}
\newtheorem*{Def*}{Definition}

\newtheorem*{question*}{Question}

\newtheorem*{problem*}{Problem}

\makeatletter
\def\thmhead@plain#1#2#3{%
  \thmname{#1}\thmnumber{\@ifnotempty{#1}{ }\@upn{#2}}%
  \thmnote{ {\the\thm@notefont#3}}}
\let\thmhead\thmhead@plain
\makeatother

\theoremstyle{definition}
\newtheorem{rem}[thm]{Remark}
\newtheorem*{note}{Note}

\newtheorem{manuallemmainner}{Lemma}
\newenvironment{manuallemma}[1]{%
  \renewcommand\themanuallemmainner{#1}%
  \manuallemmainner \it
}{\endmanuallemmainner}

\newcommand{\bb}{\begin{equation}\begin{aligned}\hspace{0pt}}
\newcommand{\bbb}{\begin{equation*}\begin{aligned}}
\newcommand{\ee}{\end{aligned}\end{equation}}
\newcommand{\eee}{\end{aligned}\end{equation*}}
\newcommand*{\coloneqq}{\mathrel{\vcenter{\baselineskip0.5ex \lineskiplimit0pt \hbox{\scriptsize.}\hbox{\scriptsize.}}} =}

\newcommand\floor[1]{\left\lfloor#1\right\rfloor}
\newcommand\ceil[1]{\left\lceil#1\right\rceil}
\newcommand{\texteq}[1]{\stackrel{\mathclap{\mbox{\scriptsize #1}}}{=}}
\renewcommand{\textleq}[1]{\stackrel{\mathclap{\mbox{\scriptsize #1}}}{\leq}}

\renewcommand{\textgeq}[1]{\stackrel{\mathclap{\mbox{\scriptsize #1}}}{\geq}}
\newcommand{\ketbra}[1]{\ket{#1}\!\!\bra{#1}}

\newcommand{\ketbraa}[2]{\ket{#1}\!\!\bra{#2}}

\newcommand{\sumno}{\sum\nolimits}
\newcommand{\e}{\varepsilon}

\newcommand{\tcb}[1]{{\color{blue} #1}}

\newcommand{\id}{\mathds{1}}
\newcommand{\R}{\mathds{R}}
\newcommand{\N}{\mathds{N}}

\newcommand{\C}{\mathds{C}}

\newcommand{\locc}{\mathrm{LOCC}}
\newcommand{\sep}{\mathds{SEP}}
\newcommand{\ppt}{\mathds{PPT}}

\newcommand{\mlocc}{\mathds{LOCC}}

\newcommand{\mk}{\mathds{K}}

\newcommand{\sq}{E_{\mathrm{sq}}}
\newcommand{\FF}{\pazocal{F}}

\newcommand{\SEP}{\mathrm{SEP}}
\newcommand{\PPT}{\mathrm{PPT}}
\newcommand{\sepp}{\mathrm{NE}}
\newcommand{\pptp}{\mathrm{PPTP}}
\newcommand{\kp}{\mathrm{KP}}

\newcommand{\SIO}{\mathrm{SIO}}
\newcommand{\IO}{\mathrm{IO}}
\newcommand{\DIO}{\mathrm{DIO}}
\newcommand{\MIO}{\mathrm{MIO}}

\DeclareMathOperator{\Tr}{Tr}

\DeclareMathOperator{\co}{conv}

\DeclareMathAlphabet{\pazocal}{OMS}{zplm}{m}{n}

\newcommand{\HH}{\pazocal{H}}

\newcommand{\Tp}{\pazocal{T}_+}
\newcommand{\D}{\pazocal{D}}
\newcommand{\K}{\pazocal{K}}

\newcommand{\MM}{\pazocal{M}}

\newcommand{\lsmatrix}{\left(\begin{smallmatrix}}
\newcommand{\rsmatrix}{\end{smallmatrix}\right)}

\stackMath
\newcommand\xxrightarrow[2][]{\mathrel{%
  \setbox2=\hbox{\stackon{\scriptstyle#1}{\scriptstyle#2}}%
  \stackunder[5pt]{%
    \xrightarrow{\makebox[\dimexpr\wd2\relax]{$\scriptstyle#2$}}%
  }{%
   \scriptstyle#1\,%
  }%
}}

\newcommand{\tendsn}[1]{\xxrightarrow[\! n\rightarrow \infty\!]{#1}}

\stackMath
\def\dyhat{-0.15ex}
\newcommand\myhat[1]{\ThisStyle{%
              \stackon[\dyhat]{\SavedStyle#1}
                              {\SavedStyle\widehat{\phantom{#1}}}}}

\makeatletter

\counterwithin*{equation}{part}
\counterwithin*{thm}{part}
\counterwithin*{figure}{part}

\tikzset{meter/.append style={draw, inner sep=10, rectangle, font=\vphantom{A}, minimum width=30, line width=.8, path picture={\draw[black] ([shift={(.1,.3)}]path picture bounding box.south west) to[bend left=50] ([shift={(-.1,.3)}]path picture bounding box.south east);\draw[black,-latex] ([shift={(0,.1)}]path picture bounding box.south) -- ([shift={(.3,-.1)}]path picture bounding box.north);}}}
\tikzset{roundnode/.append style={circle, draw=black, fill=gray!20, thick, minimum size=10mm}}
\tikzset{squarenode/.style={rectangle, draw=black, fill=none, thick, minimum size=10mm}}

\definecolor{Blues5seq1}{RGB}{239,243,255}
\definecolor{Blues5seq2}{RGB}{189,215,231}
\definecolor{Blues5seq3}{RGB}{107,174,214}
\definecolor{Blues5seq4}{RGB}{49,130,189}
\definecolor{Blues5seq5}{RGB}{8,81,156}

\definecolor{Greens5seq1}{RGB}{237,248,233}
\definecolor{Greens5seq2}{RGB}{186,228,179}
\definecolor{Greens5seq3}{RGB}{116,196,118}
\definecolor{Greens5seq4}{RGB}{49,163,84}
\definecolor{Greens5seq5}{RGB}{0,109,44}

\definecolor{Reds5seq1}{RGB}{254,229,217}
\definecolor{Reds5seq2}{RGB}{252,174,145}
\definecolor{Reds5seq3}{RGB}{251,106,74}
\definecolor{Reds5seq4}{RGB}{222,45,38}
\definecolor{Reds5seq5}{RGB}{165,15,21}

\usepackage{xr}
\makeatletter

\newcommand*{\addFileDependency}[1]{
\typeout{(#1)}
\@addtofilelist{#1}
\IfFileExists{#1}{}{\typeout{No file #1.}}
}\makeatother

\makeatletter

\newcommand{\PPTP}{\mathrm{PPTP}}
\newcommand{\FFcc}{\FF^{cc}}
\newcommand{\locccc}{\locc^{cc}}

\newcommand{\para}[1]{\smallskip\textbf{\textit{#1}}.\,---}
\newcommand{\parait}[1]{\smallskip{\textit{#1}}.\,---}

\allowdisplaybreaks

\usepackage[left=.73in,right=.73in,top=.7in,bottom=1in]{geometry}

\let\tcb\relax

\begin{document}

\title{No-go theorem for entanglement distillation using catalysis}
 
\author{Ludovico Lami}
\email{ludovico.lami@gmail.com}
\affiliation{QuSoft, Science Park 123, 1098 XG Amsterdam, the Netherlands}
\affiliation{Korteweg--de Vries Institute for Mathematics, University of Amsterdam, Science Park 105-107, 1098 XG Amsterdam, the Netherlands}
\affiliation{Institute for Theoretical Physics, University of Amsterdam, Science Park 904, 1098 XH Amsterdam, the Netherlands}

\author{Bartosz Regula}
\affiliation{Mathematical Quantum Information RIKEN Hakubi Research Team, RIKEN Cluster for Pioneering Research (CPR) and RIKEN Center for Quantum Computing (RQC), Wako, Saitama 351-0198, Japan}

\author{Alexander Streltsov}
\affiliation{Centre for Quantum Optical Technologies, Centre of New Technologies, University of Warsaw, Banacha 2c, 02-097 Warsaw, Poland}

\begin{abstract}
The use of ancillary quantum systems known as catalysts is known to be able to enhance the capabilities of entanglement transformations under local operations and classical communication. However, the limits of these advantages have not been determined, and in particular it is not known if such assistance can overcome the known restrictions on asymptotic transformation rates --- notably the existence of bound entangled (undistillable) states. Here we establish a general limitation of entanglement catalysis: we show that catalytic transformations can never allow for the distillation of entanglement from a bound entangled state with positive partial transpose, even if the catalyst may become correlated with the system of interest, and even under permissive choices of free operations. This precludes the possibility that catalysis can make entanglement theory asymptotically reversible. Our methods are based on new asymptotic bounds for the distillable entanglement and entanglement cost assisted by correlated catalysts. 
\end{abstract}

\maketitle

The study of quantum entanglement as a resource has been one of the most fundamental problems in the field of quantum information ever since its inception~\cite{horodecki_2009}.
To utilize this resource efficiently, it is often required to transform and manipulate entangled quantum systems, which leads to the well-studied question of how quantum states can be converted using only local operations and classical communication (LOCC)~\cite{bennett_1996-1,bennett_1996}.
The limits of such conversion capability are represented by asymptotic transformation rates: given many copies of an input quantum state $\rho$, how many copies of a desired target state can we obtain per each copy of $\rho$?
Such rates are particularly important in the context of purifying noisy quantum states into singlets $\Phi_2$, a task known as entanglement distillation, as well as for the reverse task of using such singlets to synthesize noisy quantum states.
This leads to the notions of \emph{distillable entanglement} $E_d(\rho)$~\cite{bennett_1996-1}, which tells us how many copies of $\Phi_2$ we can extract from a given state $\rho$, and 
of \emph{entanglement cost} $E_c(\rho)$~\cite{bennett_1996}, which tells us how many copies of a pure singlet are needed to produce $\rho$.

A phenomenon that can severely restrict our ability to extract entanglement is known as \emph{bound entanglement}~\cite{HorodeckiBound}: there exist states from which no entanglement can be distilled, even though their entanglement cost is non-zero. A consequence of this is the irreversibility of entanglement theory --- after performing a transformation $\rho \to \omega$, one may not be able to realize the reverse process $\omega \to \rho$ and recover all of the supplied copies of $\rho$. This contrasts with the asymptotic reversibility of theories such as classical and quantum thermodynamics~\cite{Popescu1997,HorodeckiBound,Brandao-thermo}. Although reversibility may still hold in some restricted cases (e.g., for all bipartite pure quantum states~\cite{bennett_1996,Popescu1997}), and there are even approaches that may enable reversibility by suitably relaxing the restrictions on the allowed physical transformations~\cite{audenaert_2003,brandao_2010,brandao_2008-1,regula_2023}, irreversibility is often a fundamental property of the theory of quantum entanglement that may not be easily evaded~\cite{lami_2023}. It is then important to understand how, if at all, irreversibility can be overcome.

A promising approach to increase the capabilities of entanglement transformations is the use of so-called \emph{catalysts}~\cite{jonathan_1999}, that is, ancillary systems that can be employed in the conversion protocol, but must eventually be returned in an unchanged state. Although this phenomenon has been shown to be remarkably powerful in the context of single- and many-copy transformations~\cite{jonathan_1999,duan_2005,klimesh_2007,turgut_2007,shiraishi_2021,kondra_2021,lipka-bartosik_2021,Takagi2022}, it is unknown whether catalysis can enhance asymptotic conversion rates. This motivates in particular an important question: is the use of catalysis enough to facilitate the reversibility of entanglement theory?

In this paper, we close this question by showing that even very permissive forms of catalytic transformations are insufficient to distill entanglement from bound entangled states. Specifically, we consider the representative class of bound entangled states known as positive partial transpose (PPT) states
and we show that the catalytically distillable entanglement of any such state is zero, which is strictly less than its catalytic entanglement cost. The result relies on the establishment of a general upper bound on distillable entanglement under catalytic LOCC operations, namely, the relative entropy of PPT entanglement, which was known to be an upper bound only in conventional, non-catalytic protocols~\cite{vedral_1998,hayashi_2006-2}. We show that this limitation persists even if one allows the catalyst to build up correlations with the main system, as well as if one allows sets of operations larger than LOCC, in particular all PPT-preserving transformations. This presents a very general limitation on the power of catalytic transformations of entangled states. We additionally study the applications of resource monotones to constraining asymptotic state conversion with catalytic assistance, obtaining a number of bounds that may be of independent interest.

\para{Preliminaries} We use $\SEP(A\!:\!B)$ to denote the set of 
states $\sigma_{AB}$ which are separable across the bipartition $A\!:\!B$. The notation $\PPT(A\!:\!B)$ will be used to denote the set of positive partial transpose states, i.e., ones for which the partially transposed operator $\sigma_{AB}^\Gamma$ is also a valid quantum state. States which are not in PPT will be
conventionally called NPT (non-positive partial transpose). 

Even though the choice of LOCC in the context of entanglement transformations is well motivated from a practical perspective, in many settings there exist other possible choices of allowed `free' operations; let us then use $\FF$ to denote the chosen set of such permitted protocols. One such choice is the set of so-called PPT operations~\cite{rains_2001}, or the even larger 
\tcb{set} of all PPT-preserving operations $\PPTP$~\cite{horodecki_2001}, 
\tcb{comprising} all maps $\Lambda : AB \to A'B'$ such that $\Lambda(\sigma_{AB}) \in \PPT(A'\!:\!B')$ for all $\sigma_{AB} \in \PPT(A\!:\!B)$. \tcb{The latter is one of the largest and most permissive sets considered in the study of operational entanglement transformations.}

Given two bipartite states $\rho_{AB}$ and $\omega_{A'B'}$, we say that the transformation from $\rho_{AB}$ to $\omega_{A'B'}$ is possible via operations in $\FF$ assisted by catalysts if there exists a finite-dimensional state $\tau_{CD}$ and an operation $\Lambda\in \FF\left(AC:BD\to A'C:B'D\right)$ such that
\bb
\Lambda\left(\rho_{AB} \otimes \tau_{CD}\right) = \omega_{A'B'} \otimes \tau_{CD}.
\ee
We denote this by $\rho_{AB} \overset{\FF^{\,c}}{\longrightarrow} \omega_{A'B'}$. 
More generally, we say that the transformation is possible via operations in $\FF$ assisted by \emph{correlated catalysts}~\cite{aberg_2014,lostaglio_2015,wilming_2017}\tcb{,} and we write $\rho_{AB} \overset{\FF^{\,cc}}{\longrightarrow} \omega_{A'B'}$, if there exists a finite-dimensional state $\tau_{CD}$ and an operation $\Lambda\in \FF\left(AC:BD\to A'C:B'D\right)$ such that
\bb
\Tr_{CD} \Lambda\left(\rho_{AB} \otimes \tau_{CD}\right) = \omega_{A'B'}
\ee
and
\bb
\Tr_{A'B'} \left[ \Lambda\left(\rho_{AB} \otimes \tau_{CD}\right) \right] = \tau_{CD}.
\ee
This relaxed notion allows for the output state of the protocol to 
\tcb{exhibit correlations} between the main system (
\tcb{$A'B'$}) and catalyst ($CD$), as long the marginal systems satisfy the required constraints.
Crucially, correlated catalysis is a strictly more powerful framework than standard catalysis, and allowing for such correlations can greatly enlarge the set of achievable state transformations 
already in the single-shot regime~\cite{lostaglio_2015,wilming_2017,muller_2018,Shiraishi2021,Kondra2021,rethinasamy_2020,wilming_2021,Takagi2022}.

Given any allowed choice of transformations $\widetilde{\FF} \in \{ \FF, \FF^c, \FFcc \}$, we write $\rho_{AB} \overset{\widetilde{\FF}}{\longrightarrow}\, \approx_\e\! \omega_{A'B'}$ if there exists a state $\omega'_{A'B'}$ such that the transformation is realizable up to some small error $\e$, which we quantify with the trace distance:
\bb
\rho_{AB} \overset{\widetilde{\FF}}{\longrightarrow} \omega'_{A'B'}\, ,\qquad \frac12\left\| \omega'_{A'B'} - \omega_{A'B'} \right\|_1\leq \e\, .
\ee
\tcb{A conceptually simplified picture can be obtained by looking at the \emph{ultimate} limitations to which the above transformation is subjected. This intuition can be formalized by investigating the conversion of a large number $n$ of copies of $\rho_{AB}$ into as many copies as possible of $\omega_{A'B'}$, with the transformation error vanishing with growing $n$. The relevant figure of merit is then the transformation rate, i.e., the ratio between the number of output copies and the number of input copies (that is, $n$). Mathematically, this can be defined by}
\bb
&R_{\widetilde{\FF}}\left( \rho_{AB} \to \omega_{A'B'} \right) \\
&\coloneqq \sup\left\{ R:\ \rho_{AB}^{\otimes n} \overset{\widetilde{\FF}}{\longrightarrow}\, \approx_{\e_n}\! \omega_{A'B'}^{\otimes \ceil{Rn}},\quad \lim_{n\to\infty} \e_n = 0 \right\} .
\ee
The distillable entanglement and entanglement cost under operations in $\widetilde{\FF}$ are then defined by
\bb
E_{d,\, \widetilde{\FF}}\left(\rho\right) \coloneqq R_{\widetilde{\FF}}\left( \rho \to \Phi_2 \right) ,\quad E_{c,\, \widetilde{\FF}}\left(\rho\right) \coloneqq \frac{1}{R_{\widetilde{\FF}}\left( \Phi_2\to \rho\right)}\, ,
\label{distillable_and_cost}
\ee
where $\Phi_2 \coloneqq \ketbra{\Phi_2}$ denotes the maximally entangled two-qubit state, $\ket{\Phi_2} = \frac{1}{\sqrt{2}} (\ket{00}+\ket{11})$. 
Conventionally, the notation $E_d$ and $E_c$ is used to refer to $E_{d,\, \rm LOCC}$ and $E_{c,\, \rm LOCC}$. 

An entangled state $\sigma_{AB}$ is called bound entangled if $E_{d,\locc} (\sigma_{AB}) = 0$. A particularly useful criterion to detect undistillability was established in~\cite{HorodeckiBound}: if a state $\sigma_{AB}$ is PPT, then $E_{d,\locc}(\sigma_{AB}) = 0$. As $E_{c,\locc}(\sigma_{AB}) > 0$ for any entangled state $\sigma_{AB}$~\cite{yang_2005}, this means that any PPT $\sigma_{AB}$ which is not separable has a non-zero entanglement cost, while no entanglement can be extracted from it. Interestingly, it is still an open question whether \emph{every} bound entangled state is PPT~\cite{Horodecki-open-problems,OpenProblem-NPT-bound}. 

\para{Monotones} A very common way to constrain entanglement transformations, also in the asymptotic transformation regime, is to use so-called entanglement monotones, also known as entanglement measures~\cite{vedral_1997}. These are functions $M$ which satisfy $M(\Lambda(\rho_{AB})) \leq M(\rho_{AB})$ for all free operations $\Lambda \in \FF$. It is well known that, if the monotone satisfies weak additivity, i.e., $M(\rho^{\otimes n}) = n M(\rho)$, as well as a stronger form of continuity known as asymptotic continuity~\cite{donald_1999,synak-radtke_2006}, then the (non-catalytic) transformation rate is bounded as~\cite{donald_2002,horodecki_2001}
\bb
R_{\FF}(\rho \to \omega) \leq \frac{M(\rho)}{M(\omega)}.
\ee
Monotones are typically chosen so that they are normalized on the maximally entangled state, i.e., $M(\Phi_2) = 1$. Any such monotone then satisfies~\cite{Horodecki2000}
\bb
E_{d,\FF}(\rho) \leq M(\rho) \leq E_{c,\FF}(\rho).
\ee

\tcb{%
A particularly important example of an LOCC monotone is the relative entropy of (NPT) entanglement~\cite{vedral_1997}
\bb
D_{\PPT} (\rho) \coloneqq \min_{\sigma \in \PPT(A:B)}  D(\rho \| \sigma)
\ee
with the quantum relative entropy defined by $D(\omega\|\tau)\coloneqq \Tr \omega \left(\log_2\omega - \log_2\tau\right)$~\cite{umegaki_1962,hiai_1991}. However, this measure is not additive: it is merely \emph{sub-additive}, in the sense that
\bb
D_{\PPT}(\rho_{A:B} \otimes \omega_{A':B'}) \leq D_{\PPT}(\rho_{A:B}) + D_{\PPT}(\omega_{A':B'}),
\label{PPT_relent}
\ee
and the inequality may be strict for some states~\cite{vollbrecht_2001}.
This issue is circumvented through the procedure of \emph{regularization}, which considers the asymptotic limit
\bb
D^\infty_{\PPT} (\rho) \coloneqq \lim_{n\to\infty} \frac1n D_{\PPT}\big(\rho^{\otimes n}\big).
\ee
The resulting regularized relative entropy of entanglement is weakly additive~\cite{donald_2002} and constitutes one of the most fundamental and commonly used bounds on entanglement transformation rates: in particular, $E_{d,\locc}(\rho) \leq D^\infty_{\PPT} (\rho)$.
}%

The situation is much more intricate when it comes to catalytic transformations~\cite{kondra_2021,rubboli_2022,catalysis-review}. To establish a similar bound, it appears that several more assumptions about the given monotone are needed. In particular, if we also assume full additivity (i.e., $M(\rho_{AB} \otimes \omega_{A'B'}) = M (\rho_{AB}) + M(\omega_{A'B'})$ for any $\rho_{AB},\omega_{A'B'}$) and strong superadditivity (i.e., $M(\rho_{AA':BB'}) \geq M(\rho_{AB}) + M(\rho_{A'B'})$), then we analogously obtain
\bb
R_{\FFcc}(\rho \to \omega) \leq \frac{M(\rho)}{M(\omega)}
\ee
(see~\cite{NoteX} 
for a proof). However, to date, there are only two LOCC monotones that are known to satisfy all the required assumptions: the squashed entanglement $E_{\rm sq}$~\cite{christandl_2004} and the conditional entanglement of mutual information $E_{I}$~\cite{Yang2007, Yang2008}, both of which are 
however typically difficult to evaluate or even 
\tcb{estimate.}

Importantly, as the regularized relative entropy $D^\infty_{\PPT}$ is not known to satisfy the above required properties, 
we do not yet know whether it is monotone under asymptotic correlated--catalytic protocols. This entails that we cannot straightforwardly use it to bound $E_{d,\locccc}$ or $E_{c,\locccc}$, and thus to alleviate the issue of the lack of readily computable bounds on transformation rates.
Any attempt to derive such asymptotic bounds on transformations with correlated catalysts therefore requires a completely different approach than conventional, non-catalytic bounds.

\para{Results} 
Our main technical contribution is the establishment of very general bounds on correlated catalytic transformation rates, and in particular the recovery of the regularized relative entropy as an upper bound on the rate of distillation:
\tcb{%
\bb
\label{eq:upper_bound_distillable_main}
E_{d,\, \locc^{cc}}(\rho_{AB}) \leq D^\infty_\PPT(\rho_{AB}) \leq D_\PPT(\rho_{AB}). 
\ee
As $D_{\PPT}(\rho_{AB}) = 0$ for any PPT state, a key consequence of this result is the fact that correlated catalysis is not sufficient to break the fundamental irreversibility of PPT bound entangled states.

\begin{thm} \label{PPT_bound_thm_main}
All PPT entangled states $\rho_{AB}$ are bound entangled under LOCC operations assisted by correlated catalysts, but have non-zero cost. More formally, if $\rho_{AB}$ is PPT entangled then
\bb
E_{d,\, \locc^{cc}}(\rho_{AB}) = 0 < E_{c,\, \locc^{cc}}(\rho_{AB})\, .
\ee
Consequently, entanglement theory is irreversible even under LOCC assisted by correlated catalysts.
\end{thm}

The result can be strengthened in several ways. First, our methods immediately apply also to PPT-preserving operations assisted by correlated catalysis, showing that even under such extended transformations no entanglement can be extracted from PPT states: $E_{d,\, \PPTP^{cc}}(\rho_{AB}) = 0$. 

Even more strongly, we can show that a PPT state 
\tcb{can never} be converted to an NPT state by means of PPT-preserving operations assisted by correlated catalysts, including all catalytic LOCC protocols.
This implies in particular that not only is the rate of distillation equal to zero, but not even a single copy of $\Phi_2$ can be distilled with error $\e < 1/2$,
no matter how many copies of a given PPT entangled state are at our disposal.

This conclusively shows that bound entanglement, and thus the irreversibility of entanglement theory, cannot be circumvented or even alleviated through the use of catalysts.
}%

Let us remark here that a different notion of `catalytic irreversibility' was previously considered in the seminal work of Vidal and Cirac~\cite{vidal_2001}. However, the transformations considered there are much more restricted than the ones allowed in our approach --- indeed, they are not truly `catalytic' in the sense that the preservation of the assisting ancillary system is not actually required, and furthermore no correlations are permitted between the main and the ancillary systems. 
Our setting is thus strictly more general than that of~\cite{vidal_2001}, and as far as we know it is not possible to retrieve our findings on catalytic bound entanglement using results from~\cite{vidal_2001} only.

\para{Proof idea} A crucial ingredient in our proofs is the \emph{measured relative entropy of entanglement} $D_\PPT^\ppt$, which belongs to a family of entanglement measures first studied by Piani in a pioneering work~\cite{piani_2009-1}. 
\tcb{%
It is defined as
\bb
D_\PPT^\ppt(\rho) \coloneqq \inf_{\sigma\in \PPT} \sup_{\MM \in \ppt} D\big( \MM(\rho)\, \big\|\, \MM(\sigma)\big),
\label{PPT_Piani}
\ee
where $\ppt$ denotes the set of PPT measurements --- that is, POVMs $\{M_i\}_i$ such that each operator $M_i$ is PPT, and $\MM(\rho) = \sum_i \Tr (M_i\, \rho) \ketbraa{i}{i}$ is the corresponding measurement channel.
The difference between this quantity and the relative entropy of entanglement $D_{\PPT}$ is that the relative entropy is evaluated not between quantum states, but rather the probability distributions constituted by the measurement outcomes.

}

\tcb{While~\eqref{PPT_Piani} seemingly adds a further layer of complication to~\eqref{PPT_relent}, it is in many respects a more natural and well-behaved quantity. Most importantly for us,} 
the measured relative entropy 
satisfies strong superadditivity, and in fact it allows for the establishment of a superadditivity
\tcb{-}like relation for the relative entropy of entanglement $D_\PPT$ itself: it holds that~\cite{piani_2009-1}
\bb
D_\PPT\left(\rho_{AA':BB'}\right) \geq D_\PPT(\rho_{A:B}) + D_\PPT^\ppt(\rho_{A':B'}).
\label{Piani_inequality_main}
\ee
This remarkable relation allows us to avoid having to rely solely on the properties of $D_\PPT$, which --- as we discussed before --- are not sufficient to use this quantity in the catalytic setting. 

Let us then derive the upper bound on catalytic distillable entanglement announced in Eq.~\eqref{eq:upper_bound_distillable_main}.

Assume that $R$ is an achievable rate for entanglement distillation under operations in $\PPTP^{cc}$, that is, that there exists a sequence of catalysts $\tau_n = (\tau_n)_{C_nD_n}$ on the finite-dimensional systems $C_n D_n$, and a sequence of operations $\Lambda_n \in \PPTP\left(A^n C_n : B^n 
D_n\to A_0^{\ceil{Rn}} C_n : B_0^{\ceil{Rn}} 
D_n \right)$, with $A_0$ and $B_0$ being single-qubit systems, such that
\bb
\e_n\coloneqq&\ \frac12 \left\|\Tr_{C_n D_n} \Lambda_n\left(\rho_{AB}^{\otimes n} \otimes \tau_n\right) - \Phi_2^{\otimes \ceil{Rn}} \right\|_1 
\tcb{\tendsn{} 0}\, , \\[1ex]
\tau_n =&\ \Tr_{A_0^{\ceil{Rn}} B_0^{\ceil{Rn}}} \left[ \Lambda_n\left(\rho_{AB}^{\otimes n} \otimes \tau_n\right) \right] . 
\ee
Then
\begin{align}
&n \, D_\PPT\big(\rho_{AB}\big) + D_\PPT\big(\tau_n\big) \nonumber \\
&\textgeq{(i)} D_\PPT\big(\rho_{AB}^{\otimes n}\big) + D_\PPT\big(\tau_n\big) \nonumber \\
&\textgeq{(ii)} D_\PPT\big(\rho_{AB}^{\otimes n}\otimes \tau_n\big) \nonumber \\
&\textgeq{(iii)} D_\PPT\big(\Lambda_n\big(\rho_{AB}^{\otimes n} \otimes \tau_n\big)\big) \nonumber \\
&\textgeq{(iv)} D_\PPT^\ppt\Big(\Tr_{C_n D_n} \Lambda_n\big(\rho_{AB}^{\otimes n} \otimes \tau_n\big)\Big) \nonumber \\
&\quad + D_\PPT\Big(\Tr_{A_0^{\ceil{Rn}}B_0^{\ceil{Rn}}} \Lambda_n\big(\rho_{AB}^{\otimes n} \otimes \tau_n\big)\Big) \label{eq:bigproof} \\
&= D_\PPT^\ppt\Big(\Tr_{C_n D_n} \Lambda_n\left(\rho_{AB}^{\otimes n} \otimes \tau_n\right)\Big) + D_\PPT\big(\tau_n\big) \nonumber \\
&\textgeq{(v)} D_\PPT^\ppt\Big(\Phi_2^{\otimes \ceil{Rn}}\Big) - \e_n \ceil{Rn} - g(\e_n) + D_\PPT\big(\tau_n\big) \nonumber \\
&\textgeq{(vi)} \ceil{Rn} - 1 - \e_n \ceil{Rn} - g(\e_n) + D_\PPT\big(\tau_n\big) \, . \nonumber 
\end{align}
Here: (i) and (ii)~are applications of the tensor sub-additivity of $D_\PPT$~\cite{vedral_1998}; (iii)~comes from its monotonicity under PPT-preserving operations; (iv)~descends from Piani's superadditivity--like inequality in \eqref{Piani_inequality_main}; (v)~follows from asymptotic continuity~\cite{li_2014-1,Schindler2023}, which states that~\cite{NoteX}
\bb
\left| D_\PPT^\ppt\big(\rho_{AB}\big) - D_\PPT^\ppt\big(\omega_{AB}\big) \right| \leq \e \log d + g(\e)
\ee
holds for all pairs of states $\rho_{AB}, \omega_{AB}$ at trace distance $\e\coloneqq \frac12 \left\|\rho_{AB} - \omega_{AB}\right\|_1$ on all systems $AB$ of minimal local dimension $d\coloneqq \min\{|A|,|B|\}$, where $g(x)\coloneqq (1+x)\log(1+x) - x\log x$; and finally~(vi) is a consequence of the quasi-normalization of $D^\ppt_\PPT$, i.e., the fact that $D_\PPT^\ppt (\Phi_2^{\otimes k}) \geq \log (2^k+1) - 1$~\cite{li_2014-1}. The above chain of inequalities shows that
\bb
D_\PPT\big( \rho_{AB}^{\otimes n} \big) \geq (1-\e_n) \ceil{Rn} - 1 - g(\e_n)\,,
\ee
and by dividing by $n$, taking the limit as $n\to\infty$, and subsequently the supremum over all achievable rates $R$, we obtain the claimed result.

\tcb{%
To complete the proof of Theorem~\ref{PPT_bound_thm_main}, it suffices to use the fact that the squashed entanglement $E_{\rm sq}$ --- which, as we remarked before and discuss in more detail in~\cite{NoteX}, lower bounds the correlated
\tcb{-}catalytic entanglement cost --- is non-zero for any entangled state~\cite{faithful,li_2014-1, VV-Markov, berta_2023, berta_2022}.

An approach very similar to the above can be used to derive a corresponding lower bound on the entanglement cost under PPT-preserving operations assisted by catalysts, leveraging once again Piani's super-additivity relation~\eqref{Piani_inequality_main}. The chain of inequalities in this case reads
\begin{align}
&\ceil{Rn} + D_\PPT\big(\tau_n\big)\\
&= D_\PPT\Big(\Phi_2^{\otimes \ceil{Rn}}\Big) + D_\PPT\big(\tau_n\big) \nonumber \\
&\geq D_\PPT\Big(\Phi_2^{\otimes \ceil{Rn}} \otimes \tau_n\Big) \nonumber \\
&\geq D_\PPT\Big(\Lambda_n\Big(\Phi_2^{\otimes \ceil{Rn}} \otimes \tau_n\Big)\Big) \nonumber \\
&\geq D_\PPT^\ppt\Big(\Tr_{C_n D_n} \Lambda_n\Big(\Phi_2^{\otimes \ceil{Rn}} \otimes \tau_n\Big)\Big)  \\
&\quad + D_\PPT\Big(\Tr_{A_0^{\ceil{Rn}}B_0^{\ceil{Rn}}} \Lambda_n\Big(\Phi_2^{\otimes \ceil{Rn}} \otimes \tau_n\Big)\Big) \\
&= D_\PPT^\ppt\Big(\Tr_{C_n D_n} \Lambda_n\Big(\Phi_2^{\otimes \ceil{Rn}} \otimes \tau_n\Big)\Big) + D_\PPT\big(\tau_n\big) \nonumber \\
&\geq D_\PPT^\ppt\big(\rho_{AB}^{\otimes n}\big) - \e_n \log\left( d^n \right) - g(\e_n) + D_\PPT\big(\tau_n\big)\, . \nonumber
\end{align}
Eliminating $D_\PPT(\tau_n)$ on both sides, dividing by $n$, and taking the limit $n\to\infty$ shows that any rate of dilution $R$ is lower bounded by the regularization of $D^\ppt_\PPT$.

Our technical contributions derived above can be summarized as follows.
\begin{prop} \label{bounds_PPTP}
For all states $\rho_{AB}$, the distillable entanglement and \tcb{the} entanglement cost under PPT-preserving operations assisted by correlated catalysts satisfy \tcb{that}
\bb
E_{d,\, \PPTP^{cc}}(\rho_{AB}) \leq D_\PPT^\infty(\rho_{AB}) \leq D_\PPT(\rho_{AB}) 
\label{bound_distillable_NE''_main}
\ee
and
\bb
E_{c,\, \PPTP^{cc}}(\rho_{AB}) \geq D_\PPT^{\ppt,\infty}(\rho_{AB}) \geq D_\PPT^{\ppt}(\rho_{AB})\,.
\label{bound_cost_NE''_main}
\ee
\end{prop}
This gives two very general limitations on asymptotic transformation rates with correlated--catalytic assistance, notably ones that can be efficiently computed or bounded as they do not require regularization. \tcb{For example, a} 
consequence of Proposition~\ref{bounds_PPTP} coupled with the faithfulness of $D^\ppt_\PPT$~\cite{piani_2009} is that the entanglement cost of any NPT entangled state is non-zero, even under PPT-preserving operations assisted by correlated catalysis.

\tcb{A peculiarity of the bounds in~\eqref{bound_distillable_NE''_main}--\eqref{bound_cost_NE''_main} is that they do not immediately imply that 
\tcb{$E_{d,\, \FF^{cc}}(\rho_{AB}) \textleq{?} E_{c,\,\FF^{cc}}(\rho_{AB})$}
for the class of operations $\FF = \PPTP$. 

\tcb{Such an inequality} is 
central to the logical consistency of the theory, because it tells us that no net entanglement can be generated in a cycle of dilution and distillation of $\rho_{AB}$, forbidding the existence of a `perpetuum mobile' in entanglement theory~\cite{Popescu1997}. 
\tcb{This is essentially a technicality stemming from our use of $\PPTP^{cc}$ operations.} For the more operationally grounded classes of free operations $\FF=\locc, \PPT$, \tcb{the inequality can be shown (cf.~\cite{kondra_2023}), and we include a complete proof of this} 
that works for all PPT or distillable states when $\FF=\locc$, and for all states when $\FF=\PPT$. 
}

Readers interested in 
a more detailed exposition of the properties of catalytic entanglement monotones 
are encouraged to consult the Supplemental Material~\cite{NoteX}. Therein, we also consider slightly strengthened and generalized variants of the result of Theorem~\ref{PPT_bound_thm_main}.
}

\tcb{%
\para{Extension to quantum 
coherence} Quantum coherence is another important example of a quantum resource, and interestingly it shares many similarities with entanglement theory~\cite{aberg_2006,baumgratz_2014,winter_2016,streltsov_2017}. In this context, incoherent operations (IO)~\cite{baumgratz_2014,winter_2016} have emerged as the main example of a set of operations that are sufficiently powerful to allow for generic coherence distillation, yet not powerful enough to enable full reversibility. It is natural to ask whether one could improve either distillation or dilution under IO via catalysis. Extending our approach from entanglement theory, we can answer this question in the negative in the most general sense: neither the IO distillable coherence nor the IO coherence cost of any state are affected by the introduction of catalysts.
As this is beyond the scope of the entanglement--focused discussion in the paper, a thorough consideration of this setting can be found in the Supplemental Material~\cite{NoteX}.
}%

\para{Discussion} We have established general limitations on asymptotic entanglement transformation rates with correlated catalysis, precluding the possibility of using catalysis to distill entanglement from PPT states. 

Although our methods lead to robust and general constraints on the power of catalytic conversion protocols, there are still many open questions in this context. In particular, complementing the no-go results obtained here, is there \emph{any} transformation whose rate can be improved by allowing correlated catalysts?
\tcb{Furthermore, since we have shown that correlated catalysis on its own is not enough to enable the reversibility of entanglement theory, the big open question~\cite{OpenProblem,berta_2022} remains: what does it take to achieve reversibile entanglement transformations? 
}%
One context in which our approach is not able to rule out reversibility is the use of non-entangling protocols~\cite{brandao_2010,lami_2023} with catalytic assistance, making it an interesting question to investigate such a possibility.

\parait{Note} Recently, a complementary question has been independently studied in~\cite{concurrent}: given that an entangled state $\rho$ is distillable, i.e., $E_{d,\locc}(\rho) > 0$, can its distillable entanglement be increased by catalytic LOCC protocols? The question is answered in the negative using different methods. Since the results of~\cite{concurrent} apply only to distillable states, they cannot be used to derive the results presented in this manuscript. 

\parait{Acknowledgments} We are indebted to Mario Berta, Marco Tomamichel, and Andreas Winter for many enlightening discussions. We thank \tcb{Ryuji Takagi}, Thomas Theurer, and Henrik Wilming for comments on the manuscript. We are grateful to the organizers of the workshop ``Quantum resources: from mathematical foundations to operational characterisation'', held in Singapore 5--8 December 2022, during which many of the discussions leading to this work took place. A.S. was supported by the ``Quantum Coherence and Entanglement for Quantum Technology'' project, carried out within the First Team programme of the Foundation for Polish Science co-financed by the European Union under the European Regional Development Fund.

\bibliographystyle{apsc}
\bibliography{main, biblio}

\clearpage

\onecolumngrid

\newgeometry{left=1.2in,right=1.2in,top=.7in,bottom=1in}


\stepcounter{part}
\counterwithin{thm}{part}
\counterwithin{equation}{part}

\renewcommand{\theequation}{S\arabic{equation}}
\renewcommand{\thethm}{S\arabic{thm}}

\hypertarget{supp}{}

\section*{Supplemental Material}

\tableofcontents

\subsection{Catalytic transformations and monotones} \label{intro_sec}

\begin{Def} \label{catalysts_def}
Let $\FF\supseteq \locc$ be a class of free operations for entanglement theory. Given two bipartite states $\rho_{AB}$ and $\omega_{A'B'}$, we say that the transformation from $\rho_{AB}$ to $\omega_{A'B'}$ is:
\begin{enumerate}[(i)]
\item possible via operations in $\FF$, and we write $\rho_{AB} \overset{\FF}{\longrightarrow} \omega_{A'B'}$, if there exists an operation $\Lambda\in \FF\left(A:B\to A':B'\right)$ such that
\bb
\Lambda\left(\rho_{AB} \right) = \omega_{A'B'}\, ;
\ee
\item possible via operations in $\FF$ assisted by catalysts,  
and we write $\rho_{AB} \overset{\FF^{c}}{\longrightarrow} \omega_{A'B'}$, if there exists a finite-dimensional state $\tau_{CD}$ and an operation $\Lambda\in \FF\left(AC:BD\to A'C:B'D\right)$ such that
\begin{align}
\Lambda\left(\rho_{AB} \otimes \tau_{CD}\right) = \omega_{A'B'} \otimes \tau_{CD}\, ;
\end{align}
\item possible via operations in $\FF$ assisted by correlated catalysts, 
and we write $\rho_{AB} \overset{\FF^{cc}}{\longrightarrow} \omega_{A'B'}$, if there exists a finite-dimensional state $\tau_{CD}$ and an operation $\Lambda\in \FF\left(AC:BD\to A'C:B'D\right)$ such that
\bb
\Tr_{CD} \Lambda\left(\rho_{AB} \otimes \tau_{CD}\right) = \omega_{A'B'}\, ,\qquad \Tr_{A'B'} \Lambda\left(\rho_{AB} \otimes \tau_{CD}\right) = \tau_{CD}\, ;
\ee
\item possible via operations in $\FF$ assisted by marginal catalysts, 
and we write $\rho_{AB} \overset{\FF^{mc}}{\longrightarrow} \omega_{A'B'}$, if there exists a finite collection of finite-dimensional states $\tau_{C_jD_j}$, $j=1,\ldots, k$, and an operation $\Lambda\in \FF\left(AC:BD\to A'C:B'D\right)$, where $C\coloneqq C_1\ldots C_k$ and $D\coloneqq D_1\ldots D_k$, such that
\bb
\Lambda\left(\rho_{AB} \otimes \tau_{CD}\right) = \omega_{A'B'}\! \otimes \tau'_{CD}\, , \quad \Tr_{\myhat{C}_j \myhat{D}_j} \!\tau'_{CD} = \tau_{C_jD_j}\ \ \forall\ j=1,\ldots, k\, ,
\ee
where $\tau_{CD}\coloneqq \bigotimes_j \tau_{C_jD_j}$ is the initial state of the overall catalyst, and $\myhat{C}_j$ (resp.\ $\myhat{D}_j$) denotes the system obtained by tracing away all $C$ systems (resp.\ all $D$ systems) except from the $j^\text{th}$ one.
\item possible via operations in $\FF$ assisted by marginal correlated catalysts, 
and we write $\rho_{AB} \overset{\FF^{mcc}}{\longrightarrow} \omega_{A'B'}$, if there exists a finite collection of finite-dimensional states $\tau_{C_jD_j}$, $j=1,\ldots, k$, and an operation $\Lambda\in \FF\left(AC:BD\to A'C:B'D\right)$, where $C\coloneqq C_1\ldots C_k$ and $D\coloneqq D_1\ldots D_k$, such that
\bb
\Tr_{CD} \Lambda\left(\rho_{AB} \otimes \tau_{CD}\right) = \omega_{A'B'}\, ,\qquad \Tr_{A'B' \myhat{C}_j \myhat{D}_j} \Lambda\left(\rho_{AB} \otimes \tau_{CD}\right) = \tau_{C_j D_j } \quad \forall\ j=1,\ldots, k\, ,
\ee
where $\tau_{CD}\coloneqq \bigotimes_j \tau_{C_jD_j}$, and $\myhat{C}_j$ (resp.\ $\myhat{D}_j$) denotes the system obtained by tracing away all $C$ systems (resp.\ all $D$ systems) except from the $j^\text{th}$ one.
\end{enumerate}
For $\widetilde{\FF}\in \left\{\FF, \FF^c, \FF^{cc}, \FF^{mc}, \FF^{mcc}\right\}$ and $\e\in [0,1]$, we write $\rho_{AB} \overset{\widetilde{\FF}}{\longrightarrow}\, \approx_\e\! \omega_{A'B'}$ if there exists a state $\omega'_{A'B'}$ such that
\bb
\rho_{AB} \overset{\widetilde{\FF}}{\longrightarrow} \omega'_{A'B'}\, ,\qquad \frac12\left\| \omega'_{A'B'} - \omega_{A'B'} \right\|_1\leq \e\, .
\ee
For every pair of states $\rho_{AB}$ and $\omega_{A'B'}$, the corresponding asymptotic rate is given by
\bb
R_{\widetilde{\FF}}\left( \rho_{AB} \to \omega_{A'B'} \right) \coloneqq \sup\left\{ R:\ \rho_{AB}^{\otimes n} \overset{\widetilde{\FF}}{\longrightarrow}\, \approx_{\e_n}\! \omega_{A'B'}^{\otimes \ceil{Rn}},\ \lim_{n\to\infty} \e_n = 0 \right\} .
\ee
The distillable entanglement and the entanglement cost under operations in $\widetilde{\FF}$ are defined by
\bb
E_{d,\, \widetilde{\FF}}\left(\rho\right) \coloneqq R_{\widetilde{\FF}}\left( \rho \to \Phi_2 \right) ,\qquad E_{c,\, \widetilde{\FF}}\left(\rho\right) \coloneqq \frac{1}{R_{\widetilde{\FF}}\left( \Phi_2\to \rho\right)}\, .
\ee
\end{Def}

\begin{rem}
The requirement that the catalyst be finite dimensional is of a purely technical nature. Although it is quite reasonable, it may be desirable to remove it. We do not know whether or not this can be done with the techniques employed here.
\end{rem}

\begin{rem}
    In this work, we concern ourselves with transformations at the level of asymptotic rates, which is how many-copy entanglement conversion is typically studied in the literature~\cite{bennett_1996-1,horodecki_2009} and which represents a very general description of the limitations of state transformations. There are also other possible ways of describing quantum state manipulation: for example, under a more stringent definition where second-order terms in the transformation error are considered, reversibility may be even harder to achieve~\cite{Kumagai2013}. We will not explicitly study such transformations here.
\end{rem}

\begin{rem}
Clearly, for all pairs of states $\rho = \rho_{AB}$ and $\omega = \omega_{A'B'}$, it holds that
\bb
R_{\FF}\left( \rho \to \omega \right) \leq R_{\FF^c}\left( \rho \to \omega \right) \leq R_{\FF^{cc}}\left( \rho \to \omega \right) ,\ R_{\FF^{mc}}\left( \rho \to \omega \right) \leq R_{\FF^{mcc}}\left( \rho \to \omega \right) ,
\label{hierarchy_rates}
\ee
where $x\leq y,z\leq w$ means that both $x\leq y\leq w$ and $x\leq z\leq w$ hold. There does not seem to exist a universal order relation between the rates under correlated catalytic transformations ($R_{\FF^{cc}}$) and those under marginal catalytic transformations ($R_{\FF^{mc}}$). As particular cases of~\eqref{hierarchy_rates} we find that
\bb
E_{d,\, \FF}\left(\rho\right) &\leq E_{d,\, \FF^{c}}\left(\rho\right) \leq E_{d,\, \FF^{cc}}\left(\rho\right),\ E_{d,\, \FF^{mc}}\left(\rho\right) \leq E_{d,\, \FF^{mcc}}\left(\rho\right) , \\[.5ex]
E_{c,\, \FF}\left(\rho\right) &\geq E_{c,\, \FF^{c}}\left(\rho\right) \geq E_{c,\, \FF^{cc}}\left(\rho\right),\ E_{c,\, \FF^{mc}}\left(\rho\right) \geq E_{c,\, \FF^{mcc}}\left(\rho\right) .
\label{hierarchy_distillable_and_cost}
\ee
\end{rem}

Let $M$ be an $\FF$-entanglement monotone, i.e.,  a non-negative function defined on all (finite-dimensional) bipartite states that is non-increasing under operations in $\FF$, in formula
\bb
M\left(\Lambda(\rho_{AB})\right) \leq M(\rho_{AB})\qquad \forall\ \Lambda \in \FF(AB\to A'B')\, .
\ee
We say that $M$ is:
\begin{enumerate}[(a)]
\item \emph{quasi-normalized}, if $M(\Phi_d) = \log d + o(\log d)$ as $d\to\infty$, where $\Phi_d$ is the maximally entangled state of local dimension $d$, namely $\Phi_d = \ketbra{\Phi_d}$ with $\ket{\Phi_d} \coloneqq \frac{1}{\sqrt{d}} \sum_{i=1}^d \ket{ii}$;
\item[(a')] \emph{normalized}, if $M(\Phi_d) = \log d$ for all $d\in \N_+$;
\item \emph{faithful}, if $M(\rho_{AB})=0$ if and only if the state $\rho_{AB}$ is separable;
\item \emph{asymptotically continuous}, if there exist two universal functions $f,g:[0,1]\to \R$ with $f(0) = g(0) = \lim_{\e\to 0} f(\e) = \lim_{\e\to 0} g(\e) = 0$ such that 
\bb
\left| M\big(\rho_{AB}\big) - M\big(\omega_{AB}\big) \right| \leq f(\e) \log d_{AB} + g(\e)
\ee
holds for all pairs of states $\rho_{AB}, \omega_{AB}$ at trace distance $\e\coloneqq \frac12 \left\|\rho_{AB} - \omega_{AB}\right\|_1$ on systems $AB$ of total dimension $d_{AB}\coloneqq |A||B|$;
\item \emph{weakly additive}, if $M\big(\rho_{AB}^{\otimes n}\big) = n M(\rho_{AB})$ holds for all bipartite states $\rho_{AB}$ and all positive integers $n$;
\item \emph{completely additive}, or \emph{tensor additive}, if $M\big(\rho_{AB}\otimes \omega_{A'B'}\big) = M(\rho_{AB}) + M(\omega_{A'B'})$ for all pairs of states $\rho_{AB}$ and $\omega_{A'B'}$, where the bipartition is with respect to the cut $AA':BB'$; and finally
\item \emph{strongly super-additive}, if $M\big( \rho_{AA':BB'} \big) \geq M(\rho_{A:B}) + M(\rho_{A':B'})$ for all states $\rho_{AA':BB'}$ on four parties, where $\rho_{A:B}$ and $\rho_{A':B'}$ are the reductions of $\rho_{AA':BB'}$, and we used colons to highlight the Alice/Bob bipartition.
\end{enumerate}
In what follows, we will always posit that $\FF$ contains the set of LOCC operations, in formula $\FF\supseteq \locc$. Since the maximally entangled state can be converted to any other state via LOCC, under this assumption quasi-normalization implies that $M(\rho_{AB}) \leq \log d + o(\log d)$ for all finite-dimensional states $\rho_{AB}$ of local dimension $d = \min\{|A|, |B|\}$.

The proof of the lemma below is somewhat standard. The interested reader can find it in Appendix~\ref{general_bounds_proof_app}. Before we state the result, we remind the reader that a function $f:\D(\HH) \to \R$ on the set of density operators $\D(\HH)$ on the Hilbert space $\HH$ is called lower semi-continuous if for all sequences $(\rho_n)_n$ converging to some $\rho\in \D(\HH)$, i.e.,  such that $\lim_{n\to\infty} \frac12 \|\rho_n-\rho\|_1 = 0$, it holds that $f(\rho) \leq \liminf_{n\to\infty} f(\rho_n)$. Note that all continuous functions, and therefore all asymptotically continuous monotones, are in particular lower semi-continuous.

\begin{lemma} \label{general_bounds_lemma}
Let $M$ be a strongly super-additive $\FF$-monotone. If $M$ is also either quasi-normalized and asymptotically continuous, or normalized and lower semi-continuous, then
\begin{align}
E_{d,\,\FF^{cc}}(\rho_{AB}) &\leq \liminf_{n\to\infty} \frac1n \sup_{\tau_{CD}} \left\{ M\left(\rho_{AB}^{\otimes n} \otimes \tau_{CD} \right) - M(\tau_{CD}) \right\} , \label{bound_distillable_ssa} \\
E_{d,\,\FF^{mcc}}(\rho_{AB}) &\leq \liminf_{n\to\infty} \frac1n \sup_{\tau_{CD} = \bigotimes_j \tau_{C_j D_j}} \left\{ M\left(\rho_{AB}^{\otimes n} \otimes \tau_{CD} \right) - \sum_j M(\tau_{C_jD_j}) \right\} , \label{bound_distillable_mcc_ssa} 
\end{align}
where all the suprema are over finite-dimensional catalysts $\tau_{CD}$. If $M$ is strongly super-additive and asymptotically continuous, then
\begin{align}
E_{c,\,\FF^{cc}}(\rho_{AB}) &\geq \frac{M^\infty(\rho_{AB})}{\limsup_{m\to\infty} \frac1m \sup_{\tau_{CD}} \left\{ M\left(\Phi_2^{\otimes m} \otimes \tau_{CD} \right) - M(\tau_{CD}) \right\}}\, , \label{bound_cost_ssa} \\
E_{c,\,\FF^{mcc}}(\rho_{AB}) &\geq \frac{M^\infty(\rho_{AB})}{\limsup_{m\to\infty} \frac1m \sup_{\tau_{CD} = \bigotimes_j \tau_{C_j D_j}} \left\{ M\left(\Phi_2^{\otimes m} \otimes \tau_{CD} \right) - \sum_j M(\tau_{C_j D_j}) \right\}}\, . \label{bound_cost_mcc_ssa}
\end{align}
where $M^\infty(\rho) \coloneqq \lim_{n\to\infty}\frac1n M\big(\rho^{\otimes n}\big)$ (and the limit exists). If $M$ is strongly super-additive and only lower semi-continuous, then the weaker bounds
\begin{align}
E_{c,\,\FF^{cc}}(\rho_{AB}) &\geq \frac{M(\rho_{AB})}{\limsup_{m\to\infty} \frac1m \sup_{\tau_{CD}} \left\{ M\left(\Phi_2^{\otimes m} \otimes \tau_{CD} \right) - M(\tau_{CD}) \right\}}\, , \label{bound_cost_ssa_unregularized} \\
E_{c,\,\FF^{mcc}}(\rho_{AB}) &\geq \frac{M(\rho_{AB})}{\limsup_{m\to\infty} \frac1m \sup_{\tau_{CD} = \bigotimes_j \tau_{C_j D_j}} \left\{ M\left(\Phi_2^{\otimes m} \otimes \tau_{CD} \right) - \sum_j M(\tau_{C_j D_j}) \right\}} \label{bound_cost_mcc_ssa_unregularized}
\end{align}
hold.

Finally, if $M$ is strongly super-additive, completely additive, normalized, and lower semi-continuous, then
\begin{equation}
E_{d,\,\FF^{c}}(\rho_{AB}) \leq E_{d,\,\FF^{cc}}(\rho_{AB}) \leq E_{d,\,\FF^{mcc}}(\rho_{AB}) \leq M(\rho_{AB}) \leq E_{c,\,\FF^{mcc}}(\rho_{AB}) \leq E_{c,\,\FF^{cc}}(\rho_{AB}) \leq E_{c,\,\FF^{c}}(\rho_{AB})\, .
\label{chain}
\end{equation}
\end{lemma}

\begin{rem}
The primary role of the condition that the catalyst be finite-dimensional is to ensure that the differences in~\eqref{bound_distillable_ssa}--\eqref{bound_cost_mcc_ssa_unregularized} are well defined, by avoiding the indeterminacy $\infty-\infty$.
\end{rem}

\begin{rem}
For additive monotones, normalization and quasi-normalization are equivalent.
\end{rem}

It was unknown for a long time whether a non-trivial entanglement monotone that satisfies at the same time quasi-normalization, strong super-additivity, asymptotic continuity, and complete additivity would exist at all. This was the state of affairs until the \emph{squashed entanglement} was constructed~\cite{Tucci1999, christandl_2004, li_2014-1, Shirokov-sq}. For a bipartite state $\rho_{AB}$, this is defined by
\bb
\sq(\rho_{AB}) \coloneqq \inf_{\rho_{ABE}} \frac12 I(A:B|E)_\rho\, ,
\ee
where the infimum runs over all extensions of $\rho_{AB}$, i.e., over all states $\rho_{ABE}$ such that $\Tr_E \rho_{ABE} = \rho_{AB}$, and $I(A\!:\!B|E)_\rho \coloneqq I(A\!:\!BE)_\rho - I(A\!:\!E)_\rho$ is the conditional mutual information, with $I(X\!:\!Y)_\omega \coloneqq D\left(\omega_{XY} \| \omega_X \otimes \omega_Y\right)$ being the mutual information, constructed~\cite{Shirokov2016} via the Umegaki relative entropy~\cite{umegaki_1962} so that it is well defined also in infinite dimensions.

A closely related entanglement measure is the \emph{conditional entanglement of mutual information} (CEMI), defined by
\bb
E_I(\rho_{AB}) \coloneqq \inf_{\rho_{AA'BB'}} \frac12 \Big( I(AA':BB')_\rho - I(A':B')_\rho \Big)\, ,
\label{cemi}
\ee
where the infimum is over all extensions of $\rho_{AB}$ on $AA'BB'$, with $A'$ and $B'$ being arbitrary systems. It is known that $E_I(\rho_{AB})\leq \sq(\rho_{AB})$ for all states $\rho_{AB}$, and it is currently an open problem to find a state where the inequality is strict~\cite[Figure~1]{li_2014-1}.

The squashed entanglement 
and CEMI are useful because 
they possess all of the above properties (a')--(f). For the squashed entanglement, properties~(a') and (d)--(f) were established in~\cite{christandl_2004}, while~(c) follows from the seminal Alicki--Fannes inequality~\cite{Alicki-Fannes, winter_2016-1}. (See also~\cite{Shirokov-sq} for a thorough discussion of the infinite-dimensional case.) Faithfulness~(b) was claimed in~\cite{faithful}, but the proof presented there is strictly speaking not valid in its original form because it makes direct use of the Brand\~ao--Plenio generalized quantum Stein's lemma~\cite{Brandao2010}, in whose proof an issue was recently found~\cite{berta_2022}. However, the argument in~\cite{faithful} can be fixed by using a more direct derivation~\cite{berta_2022}. Besides, three alternative proofs that do not rely on the 
generalized quantum Stein's lemma are available~\cite{li_2014-1, VV-Markov, berta_2023}.

As for CEMI, properties~(a'),~(c),~(d), and~(e) are established in the original works~\cite{Yang2008, Yang2007}. Property~(b) follows from the faithfulness of the squashed entanglement together with the inequality $E_I(\rho_{AB})\leq \sq(\rho_{AB})$. As for strong super-additivity~(f), it is also well known (see, e.g.,~\cite[Remark~4.3]{Lancien2017}), and indeed it can be proved in a few lines by generalizing slightly the argument given in~\cite[Proposition~1]{Yang2007}. Namely, given a state $\rho_{AA':BB'}$, for all extensions on $AA'A''BB'B''$ we have that
\bb
&I(AA'A'':BB'B'')_\rho - I(A'':B'')_\rho \\
&\qquad = I(AA'A'':BB'B'')_\rho - I(A'A'':B'B'')_\rho + I(A'A'':B'B'')_\rho - I(A'':B'')_\rho \\
&\qquad \geq 2E_I(\rho_{A:B}) + 2E_I(\rho_{A':B'})\, ,
\ee
where, precisely as in~\cite[Eq.~(5)]{Yang2008}, the last inequality follows because $\rho_{AA'A'':BB'B''}$ is an extension of $\rho_{A:B}$, and $\rho_{A'A'':B'B''}$ is an extension of $\rho_{A':B'}$. Diving by $2$ and taking the infimum over $\rho_{AA'A'':BB'B''}$ yields the claimed strong super-additivity of CEMI.

The following statement, whose special case corresponding to LOCC assisted by correlated catalysts is already implicit in~\cite[Theorem~2]{Kondra2021}, is an immediate consequence of Lemma~\ref{general_bounds_lemma} and of the above discussion.

\begin{cor} \label{sq_bound_cor}
For an arbitrary finite-dimensional bipartite state $\rho_{AB}$, the chain of inequalities
\bb
E_{d,\,\locc^{c}}(\rho_{AB}) &\leq E_{d,\,\locc^{cc}}(\rho_{AB}) \leq E_{d,\,\locc^{mcc}}(\rho_{AB}) \leq \sq(\rho_{AB}) \\
&\leq E_I(\rho_{AB}) \leq E_{c,\,\locc^{mcc}}(\rho_{AB}) \leq E_{c,\,\locc^{cc}}(\rho_{AB}) \leq E_{c,\,\locc^{c}}(\rho_{AB}) .
\ee
holds.
\end{cor}

In the above case, the assumption of finite dimensionality for $AB$ \emph{can} be removed by using the technique presented in~\cite{ferrari_2022} (see also~\cite[Theorem~4 and Remark~10]{lami_2020-1}).

\subsection{Freely measured relative entropy of resource}

Consider a finite-dimensional bipartite quantum system $AB$ with Hilbert space $\HH_{AB} = \HH_A \otimes \HH_B$. In what follows, will denote by $\K=\SEP, \PPT$ one of the two cones
\begin{align}
\SEP(A\!:\!B) &\coloneqq \co \left\{ R_A\otimes S_B:\, R_A\in \Tp(\HH_A),\, S_B\in \Tp(\HH_B) \right\} , \label{SEP} \\
\PPT(A\!:\!B) &\coloneqq \left\{ T_{AB}\in \Tp(\HH_{AB}):\ T_{AB}^\Gamma\geq 0\right\} .
\label{PPT}
\end{align}
Here, $\Tp(\HH)$ denotes the cone of positive semi-definite (trace class) operators on the Hilbert space $\HH$, while $\Gamma$ stands for the partial transposition over the second sub-system $B$. For $\K\in \{\SEP, \PPT\}$, we define
\bb
\K^1 \coloneqq \left\{ T\in \K:\ \Tr T =1\right\} .
\ee
We shall also consider the sets of measurements
\begin{align}
\sep(A\!:\!B) &\coloneqq \left\{ \left(E_i\right)_{i=1,\ldots, s}:\ s\in \N_+,\ E_i\in \SEP(A\!:\!B),\ \sum_{i=1}^s E_i = \id_{AB} \right\} , \\
\ppt(A\!:\!B) &\coloneqq \left\{ \left(E_i\right)_{i=1,\ldots, s}:\ s\in \N_+,\ E_i\in \PPT(A\!:\!B),\ \sum_{i=1}^s E_i = \id_{AB} \right\} ,
\end{align}
which are natural mathematical relaxations of the set $\mlocc(A\!:\!B)$ of LOCC-implementable measurements.

\tcb{
Two sets of operations will be of particular interest to us in addition to LOCC. These are the non-entangling (separability-preserving) operations $\sepp$ and the PPT-preserving operations $\pptp$. They are both defined as quantum channels that preserve the given cone; specifically, for $\K\in \{\SEP, \PPT\}$, a map $\Lambda : AB \to A'B'$ is $\K$-preserving if $\Lambda[\K(A:B)] \subseteq \K(A':B')$. Importantly, both sets are supersets of LOCC.

\begin{rem}
Rains~\cite{rains_2001} introduced the class of `PPT operations' as maps whose Choi operator is PPT. They are sometimes confused with PPT-preserving maps~\cite{horodecki_2001}, but PPT operations are actually those that \emph{completely} preserve the PPT cone, in the sense that $(\Lambda \otimes \mathrm{id}) \left[ \K\big( A \widetilde{A} : B \widetilde{B} \big)\right] \subseteq \K\big( A' \widetilde{A} : B' \widetilde{B} \big)$ where $\widetilde{A}, \widetilde{B}$ are arbitrary systems. It is not difficult to 
\tcb{see} that $\locc \subsetneq \PPT \subsetneq \pptp$.
\end{rem}
}%

We now introduce formally a family of entanglement measures first studied by Piani in a pioneering paper~\cite{piani_2009-1}, and subsequently investigated and generalized, among others, by Li and Winter~\cite{li_2014-1} and Brand\~{a}o, Harrow, Lee, and Peres~\cite{brandao_2020}.

\begin{Def}[(Freely measured relative entropy of entanglement)]
For $\K\in \{\SEP, \PPT\}$ and $\mk\in \{\mlocc, \sep, \ppt\}$, we define
\bb
D_\K^\mk(\rho) \coloneqq \inf_{\sigma\in \K^1} \sup_{\MM\in \mk} D\big( \MM(\rho)\, \big\|\, \MM(\sigma)\big)\, ,
\label{DKK}
\ee
where $D(p\|q) \coloneqq \sum_x p_x \log \frac{p_x}{q_x}$ is the Kullback--Leibler divergence.
\end{Def}

The above definition should be contrasted with that of the relative entropy of (NPT) entanglement, given by~\cite{vedral_1997, vedral_1998}
\bb
D_\K(\rho) \coloneqq \inf_{\sigma\in \K^1} D(\rho\|\sigma)\, ,
\label{DK}
\ee
where $D(\rho\|\sigma)\coloneqq \Tr \rho (\log\rho - \log \sigma)$ is the Umegaki relative entropy~\cite{umegaki_1962}. \tcb{Clearly, by the data processing inequality~\cite{Lindblad-monotonicity,lieb73a,lieb73b}, we have in general that~\eqref{DK} is never smaller than~\eqref{DKK}, meaning that
\bb
D_\K^\mk(\rho) \leq D_\K(\rho)
\label{data_processing_DK_DKK}
\ee
for all $\K\in \{\SEP, \PPT\}$ and $\mk\in \{\mlocc, \sep, \ppt\}$. While~\eqref{DK} may appear superficially to be a more natural entanglement measure than~\eqref{DKK}, we will argue now that this is not necessarily true. Indeed, the freely measured version of the relative entropy of entanglement obeys even more properties than its standard, un-measured version.
}
The remainder of this section is a summary of all the properties of the freely measured relative entropy of entanglement, as collected from the existing literature. For the sake of completeness, self-contained proofs of all of these results are presented in Appendix~\ref{proofs_DKK_app}.

One first elementary observation, also due to Piani~\cite{piani_2009-1}, is that $D_\K^\mk$ is a faithful measure of entanglement (when $\K=\SEP$) or NPT entanglement (when $\K=\PPT$), provided that $\mk$ is informationally complete. In mathematical terms, this means that $D_\K^\mk(\rho) = 0$ if and only if $\rho\in \K$. This can be easily seen by applying Pinsker's inequality~\cite{Pinsker} to the Kullback--Leibler divergence in~\eqref{DKK}, and observing that if $\mk$ is informationally complete the resulting supremum defines a norm, called the distinguishability norm associated with $\mk$~\cite{VV-dh, ultimate}. One arrives at the inequality~\cite{piani_2009-1}
\bb
D_\K^\mk(\rho) \geq \frac{1}{2\ln 2} \left(\inf_{\sigma\in \K^1} \left\|\rho-\sigma\right\|_\mk \right)^2 ,
\label{faithfulness_DKK}
\ee
where $\|X\|_\mk \coloneqq \sup_{\MM\in \mk} \left\|\MM(X)\right\|_1$. Since $\|\cdot\|_\mk$ is a norm and not a semi-norm under the above assumption on $\mk$, faithfulness of $D_\K^\mk$ follows immediately. We continue by recording a simple lemma, first established by Piani~\cite[Theorem~2]{piani_2009-1}.

\begin{lemma}[{\cite{piani_2009-1} (Monotonicity)}] \label{elementary_properties_DKK_lemma}
For $\K\in \{\SEP, \PPT\}$ and $\mk\in \{\mlocc, \sep, \ppt\}$, the function $D_\K^\mk$ is invariant under local unitaries, convex, and monotonically non-increasing under LOCC. When $\K=\PPT$ and $\mk = \ppt$, it is also monotonically non-increasing under PPT operations.
\end{lemma}

The following is a slight improvement over~\cite[Proposition~3]{li_2014-1}, obtained by leveraging the Alicki--Fannes--Winter technique developed in~\cite{Alicki-Fannes, winter_2016-1}. The same result has been obtained independently in~\cite[Theorem~11]{Schindler2023}.

\begin{lemma}[(Asymptotic continuity)] \label{asymptotic_continuity_DKK_lemma}
Let $AB$ be a bipartite system of minimum local dimension $d\coloneqq \min\{|A|,|B|\}<\infty$. Let $\rho,\omega$ be two states on $AB$ at trace distance $\e\coloneqq \frac12 \left\|\rho-\omega\right\|_1$. For $\K\in \{\SEP, \PPT\}$ and any set of measurements $\mk$, it holds that
\bb
\left| D_\K^\mk(\rho) - D_\K^\mk(\omega)\right| \leq \epsilon \log d + g(\epsilon)\, ,
\label{asymptotic_continuity_DKK}
\ee
where
\bb
g(x) \coloneqq (1+x)\log(1+x) - x\log x\, .
\label{g}
\ee
\end{lemma}

The following result is taken from~\cite[Proposition~4]{li_2014-1}, where it is also shown that an equality in fact holds in~\eqref{max_ent_DKK}. (We will however not make use of this latter fact.)

\begin{lemma}[{\cite[Proposition~4]{li_2014-1} (Quasi-normalization)}] \label{normalization_DKK_lemma}
Given a cone $\K\in \{\SEP, \PPT\}$ and a set of measurements $\mk\in \{\mlocc,\sep,\ppt\}$, the maximally entangled state $\Phi_d$ on a $d\times d$ bipartite quantum system satisfies that
\bb
D_\K^\mk (\Phi_d) \geq \log (d+1) - 1\, .
\label{max_ent_DKK}
\ee
\end{lemma}

An absolutely crucial property that sets apart $D_\K^\mk$ from other entanglement measures based on the relative entropy, such as the relative entropy of entanglement, is strong super-additivity. This has been first established by Piani in~\cite[Theorem~2]{piani_2009-1}, and it has been later further generalized~\cite[Lemma~15]{brandao_2020}.

\begin{lemma}[{\cite[Theorem~2]{piani_2009-1} (Strong super-additivity)}] \label{strong_superadditivity_DKK_lemma}
Let $\rho_{AA':BB'}$ be an arbitrary state over a pair of bipartite systems $AB$ and $A'B'$. For
\bb
\left(\K, \mk\right) \in \left\{ \left(\SEP, \mlocc \right), \left(\SEP, \sep \right), \left(\PPT, \mlocc \right), \left(\PPT, \sep \right), \left(\PPT, \ppt \right) \right\} ,
\label{pairs}
\ee
i.e., for all possible choices of $\K\in \{\SEP, \PPT\}$ and $\mk\in \{\mlocc,\sep,\ppt\}$ except $\K=\SEP$ and $\mk=\ppt$, it holds that
\bb
D_\K^\mk\left(\rho_{AA':BB'}\right) \geq D_\K^\mk\left(\rho_{A:B}\right) + D_\K^\mk\left(\rho_{A':B'}\right) .
\label{superadditivity_DKK}
\ee
\end{lemma}

Among other things, the above result implies, via Fekete's lemma~\cite{Fekete1923}, that the regularization
\bb
D_\K^{\mk,\infty}\left(\rho_{AB}\right) \coloneqq \lim_{n\to\infty} \frac1n D_\K^\mk\left(\rho_{AB}^{\otimes n}\right) = \sup_{n\in \N_+} \frac1n D_\K^\mk\left(\rho_{AB}^{\otimes n}\right)
\label{regularized_DKK}
\ee
is well defined for all states $\rho_{AB}$ and for all pairs $(\K,\mk)$ as in~\eqref{pairs}. For analogous but opposite reasons, since $D_\K$ is tensor \emph{sub-}additive, i.e., such that
\bb
D_\K \left(\rho_{A:B} \otimes \omega_{A':B'}\right) \leq D_\K(\rho_{A:B}) + D_\K(\omega_{A':B'})\, ,
\label{tensor_subadditivity_DK}
\ee
the regularization
\bb
D_\K^\infty \left(\rho_{AB}\right) \coloneqq \lim_{n\to\infty} \frac1n D_\K\left(\rho_{AB}^{\otimes n}\right) = \inf_{n\in \N_+} \frac1n D_\K\left(\rho_{AB}^{\otimes n}\right)
\label{regularized_DK}
\ee
is also well defined for $\K\in\{\SEP, \PPT\}$.

A variation of the argument used to prove the above Lemma~\ref{strong_superadditivity_DKK_lemma} has been employed by Piani to show an inequality related to~\eqref{superadditivity_DKK}. Since it is this inequality that will play a crucial role in our approach, rather than simple strong super-additivity~\eqref{superadditivity_DKK}, we review its proof in Appendix~\ref{proofs_DKK_app}.

\begin{lemma}[{\cite[Theorem~1]{piani_2009-1} (Piani's super-additivity-like inequality)}] \label{Piani_inequality_lemma}
Let $\rho_{AA':BB'}$ be an arbitrary state over a pair of bipartite systems $AB$ and $A'B'$. For all pairs $(\K,\mk)$ satisfying~\eqref{pairs}, i.e., for all possible choices of $\K\in \{\SEP, \PPT\}$ and $\mk\in \{\mlocc,\sep,\ppt\}$ except 
for the pair $(\SEP, \ppt)$, it holds that
\bb
D_\K\left(\rho_{AA':BB'}\right) \geq D_\K\left(\rho_{A:B}\right) + D_\K^\mk\left(\rho_{A':B'}\right) ,
\label{Piani_inequality}
\ee
where $D_\K$ is defined by~\eqref{DK}.
\end{lemma}

\subsection{Correlated catalysis and PPT-preserving or non-entangling operations}

We begin by establishing general bounds on the asymptotic rates under PPT-preserving and non-entangling operations assisted by correlated catalysts.

We stress that obtaining any asymptotic constraints in this setting is far from trivial --- the monotones that can be used to obtain restrictions on asymptotic rates under LOCC, such as squashed entanglement, are no longer known to be monotonic under larger sets of operations; moreover, more easily computable quantities such as the relative entropy of entanglement are not known to satisfy the required properties (Lemma~\ref{general_bounds_lemma}), meaning that it is not clear how to turn them into asymptotic bounds.

We will be able to circumvent these difficulties by employing Piani's super-additivity-like inequality~\eqref{Piani_inequality}.

\begin{prop} \label{bounds_NE''_prop}
For all states $\rho_{AB}$, the distillable entanglement under PPT-preserving operations assisted by correlated catalysts satisfies that
\bb
E_{d,\, \pptp^{cc}}(\rho_{AB}) \leq D_\PPT^\infty(\rho_{AB})
\label{bound_distillable_NE''}
\ee
and the correlated catalytic entanglement cost satisfies that
\bb
E_{c,\, \pptp^{cc}}(\rho_{AB}) \geq D_\PPT^{\mk,\infty}(\rho_{AB}) \geq D_\PPT^{\mk}(\rho_{AB})\, ,
\label{bound_cost_NE''}
\ee
where $\mk\in \{\mlocc, \sep, \ppt\}$, and $D_\SEP^\infty$ and $D_\SEP^{\mk,\infty}$ are defined by~\eqref{regularized_DK} and~\eqref{regularized_DKK}. In particular, $E_{c,\,\pptp^{cc}}$ is faithful.

Similarly, under non-entangling operations assisted by correlated catalysts we have that
\bb
E_{d,\, \sepp^{cc}}(\rho_{AB}) \leq D_\SEP^\infty(\rho_{AB})
\label{bound_distillable_NE''2}
\ee
and
\bb
E_{c,\, \sepp^{cc}}(\rho_{AB}) \geq D_\SEP^{\mk,\infty}(\rho_{AB}) \geq D_\SEP^{\mk}(\rho_{AB})\, ,
\label{bound_cost_NE''2}
\ee
where $\mk\in \{\mlocc, \sep\}$.
\end{prop}

\begin{rem}
We do not know at present whether the inequalities $E_{d,\,\pptp^{cc}}(\rho_{AB}) \leq E_{c,\,\pptp^{cc}}(\rho_{AB})$ and $E_{d,\,\sepp^{cc}}(\rho_{AB}) \leq E_{c,\,\sepp^{cc}}(\rho_{AB})$  hold in general. The reason is that even if a set of free operations $\FF$ is closed under concatenation, the set of \emph{approximate} state-to-state transformations induced by including assistance by correlated catalysts is not necessarily such. In other words, $\rho \overset{\FF^{cc}}{\longrightarrow}\, \approx_\e\! \omega$ and $\omega \overset{\FF^{cc}}{\longrightarrow}\, \approx_\e\! \zeta$ does not imply that $\rho \overset{\FF^{cc}}{\longrightarrow}\, \approx_{2\e}\! \zeta$. This is because the condition of exact preservation of the reduced state of the catalyst is incompatible with any $\e$-approximation.
\end{rem}

\begin{proof}[Proof of Proposition~\ref{bounds_NE''_prop}]
We will use $\kp$ to denote $\K$-preserving operations, with $\K$ standing for either the separable or the PPT cone. $\kp^{cc}$ is defined as before.

We start by proving the bound on distillable entanglement (Eq.~\eqref{bound_distillable_NE''} and \eqref{bound_distillable_NE''2}). Let $R$ be an achievable rate for entanglement distillation under operations in $\kp^{cc}$. Construct the sequence of catalysts $\tau_n = (\tau_n)_{C_nD_n}$ on the finite-dimensional systems $C_n D_n$, and the sequence of operations $\Lambda_n \in \kp\left(A^n C_n : B^n D_n \to A_0^{\ceil{Rn}} C_n : B_0^{\ceil{Rn}} D_n \right)$, with $A_0$ and $B_0$ being single-qubit systems, such that
\bb
\e_n\coloneqq&\ \frac12 \left\|\Tr_{C_n D_n} \Lambda_n\left(\rho_{AB}^{\otimes n} \otimes \tau_n\right) - \Phi_2^{\otimes \ceil{Rn}} \right\|_1 \tendsn{} 0\, , \\[1ex]
\tau_n =&\ \Tr_{A_0^{\ceil{Rn}} B_0^{\ceil{Rn}}} \left[ \Lambda_n\left(\rho_{AB}^{\otimes n} \otimes \tau_n\right) \right] . 
\ee
Then
\bb
D_\K\big(\rho_{AB}^{\otimes n}\big) + D_\K\big(\tau_n\big) &\textgeq{(i)} D_\K\big(\rho_{AB}^{\otimes n}\otimes \tau_n\big) \\
&\textgeq{(ii)} D_\K\big(\Lambda_n\big(\rho_{AB}^{\otimes n} \otimes \tau_n\big)\big) \\
&\textgeq{(iii)} D_\K^\mk\Big(\Tr_{C_n D_n} \Lambda_n\big(\rho_{AB}^{\otimes n} \otimes \tau_n\big)\Big) + D_\K\Big(\Tr_{A_0^{\ceil{Rn}}B_0^{\ceil{Rn}}} \Lambda_n\big(\rho_{AB}^{\otimes n} \otimes \tau_n\big)\Big) \\
&= D_\K^\mk\Big(\Tr_{C_n D_n} \Lambda_n\left(\rho_{AB}^{\otimes n} \otimes \tau_n\right)\Big) + D_\K\big(\tau_n\big) \\
&\textgeq{(iv)} D_\K^\mk\Big(\Phi_2^{\otimes \ceil{Rn}}\Big) - \e_n \ceil{Rn} - g(\e_n) + D_\K\big(\tau_n\big) \\
&\textgeq{(v)} \ceil{Rn} - 1 - \e_n \ceil{Rn} - g(\e_n) + D_\K\big(\tau_n\big) \, .
\ee
Here, (i)~is an application of the tensor sub-additivity of $D_\K$ as formalized in~\eqref{tensor_subadditivity_DK}, (ii)~comes from its monotonicity under non-entangling operations, (iii)~from Piani's super-additivity-like inequality in Lemma~\eqref{Piani_inequality_lemma}, (iv)~from asymptotic continuity (Lemma~\ref{asymptotic_continuity_DKK_lemma}), and finally~(v) from quasi-normalization (Lemma~\ref{normalization_DKK_lemma}). The above chain of inequalities shows that
\bb
D_\K\big( \rho_{AB}^{\otimes n} \big) \geq (1-\e_n) \ceil{Rn} - 1 - g(\e_n)\, .
\ee
Once again, the claim follows by dividing by $n$, taking the limit as $n\to\infty$, and subsequently the supremum over all achievable rates $R$.

The proof of~\eqref{bound_cost_NE''} and \eqref{bound_cost_NE''2} is entirely analogous, and leverages once again Piani's super-additivity-like inequality (Lemma~\ref{Piani_inequality_lemma}). The chain of inequalities in this case reads
\begin{align}
\ceil{Rn} + D_\K\big(\tau_n\big) &= D_\K\Big(\Phi_2^{\otimes \ceil{Rn}}\Big) + D_\K\big(\tau_n\big) \nonumber \\
&\geq D_\K\Big(\Phi_2^{\otimes \ceil{Rn}} \otimes \tau_n\Big) \nonumber \\
&\geq D_\K\Big(\Lambda_n\Big(\Phi_2^{\otimes \ceil{Rn}} \otimes \tau_n\Big)\Big) \nonumber \\
&\geq D_\K^\mk\Big(\Tr_{C_n D_n} \Lambda_n\Big(\Phi_2^{\otimes \ceil{Rn}} \otimes \tau_n\Big)\Big) + D_\K\Big(\Tr_{A_0^{\ceil{Rn}}B_0^{\ceil{Rn}}} \Lambda_n\Big(\Phi_2^{\otimes \ceil{Rn}} \otimes \tau_n\Big)\Big) \\
&= D_\K^\mk\Big(\Tr_{C_n D_n} \Lambda_n\Big(\Phi_2^{\otimes \ceil{Rn}} \otimes \tau_n\Big)\Big) + D_\K\big(\tau_n\big) \nonumber \\
&\geq D_\K^\mk\big(\rho_{AB}^{\otimes n}\big) - \e_n \log\left( d^n \right) - g(\e_n) + D_\K\big(\tau_n\big)\, . \nonumber
\end{align}
Eliminating $D_\K(\tau_n)$ on both sides, dividing by $n$, taking first the limit $n\to\infty$ and then the supremum over all achievable rates $R$ yields~\eqref{bound_cost_NE''}. Faithfulness of $E_{c,\,\kp^{cc}}$ follows from that of $D_\K^\mk$ as established by~\eqref{faithfulness_DKK}.
\end{proof}

\subsection{Irreversibility of entanglement under LOCC assisted by correlated catalysts}

Combining the insights of the previous sections, we can derive a general statement about the asymptotic irreversibility of entanglement theory assisted by correlated catalysis as follows.

\begin{thm} \label{PPT_bound_thm}
The following holds:
\begin{enumerate}[(a)]
\item A PPT state cannot be converted to an NPT state by means of PPT-preserving operations assisted by correlated catalysts.
\item In particular, not even a single ebit can be distilled with error $\e < 1/2$ by an unbounded number of copies of any given PPT state via LOCC or PPT-preserving operations assisted by correlated catalysts. 
\item Therefore, all PPT entangled states $\rho_{AB}$ are bound entangled under LOCC or PPT-preserving operations assisted by correlated catalysts, but have non-zero cost under LOCC assisted by correlated catalysts. More formally, if $\rho_{AB} \in \PPT(A\!:\!B) \setminus \SEP(A\!:\!B)$ then
\bb
E_{d,\, \locc^{cc}}(\rho_{AB}) = E_{d,\, \pptp^{cc}}(\rho_{AB}) = 0\, ,\qquad E_{c,\, \locc^{cc}}(\rho_{AB}) > 0\, .
\ee
\item Consequently, entanglement theory is irreversible even under LOCC assisted by correlated catalysts.
\end{enumerate}
\end{thm}

\begin{rem}
 The irreversibility of entanglement theory as defined in claim (c) can be deduced directly from Proposition~\ref{bounds_NE''_prop} --- we have that $E_{d, \pptp^{cc}}(\rho_{AB}) \leq D_{\PPT}(\rho_{AB}) = 0$ for a PPT state $\rho_{AB}$, while $\sq(\rho_{AB}) > 0$ as long as $\rho_{AB} \notin \SEP(A\!:\!B)$ due to the faithfulness of squashed entanglement. However, our claims (a) and (b) strengthen the irreversibility and clarify that it cannot be easily circumvented, which is why we include the Theorem as a separate result and provide a complete proof below.
\end{rem}

\begin{proof}
Let us start with claim~(a). By Corollary~\ref{sq_bound_cor}, if $\rho_{AB}$ is a generic entangled state then $E_{c,\, \locc^{cc}}(\rho_{AB}) \geq \sq(\rho_{AB}) > 0$, where the latter inequality comes from the faithfulness of the squashed entanglement. On the other hand, assume that $\rho_{AB} \overset{\pptp^{cc}}{\longrightarrow} \omega_{A'B'}$, so that there exists a catalyst $\tau_{CD}$ and an operation $\Lambda\in \pptp\left(AC:BD\to A'C:B'D\right)$ with the properties that
\bb
\Tr_{CD} \Lambda\left( \rho_{AB}\otimes \tau_{CD}\right) = \omega_{A'B'}\, ,\qquad \Tr_{AB} \Lambda\left( \rho_{AB}\otimes \tau_{CD}\right) = \tau_{CD}\, .
\label{PPT_bound_thm_proof_eq0}
\ee
Then
\bb
D_\PPT\left(\rho_{AB} \otimes \tau_{CD}\right) &\textgeq{(i)} D_\PPT\left(\Lambda\left(\rho_{AB} \otimes \tau_{CD}\right)\right) \\[.8ex]
&\textgeq{(ii)} D_\PPT^\ppt\left(\Tr_{CD}\Lambda\left(\rho_{AB} \otimes \tau_{CD}\right)\right) + D_\PPT\left(\Tr_{AB} \Lambda\left(\rho_{AB} \otimes \tau_{CD}\right)\right) \\[.8ex]
&\texteq{(iii)} D_\PPT^\ppt\left(\omega_{A'B'}\right) + D_\PPT\left(\tau_{CD}\right) .
\label{PPT_bound_thm_proof_eq1}
\ee
Here, (i)~is an application of the monotonicity of $D_\PPT$ under PPT-preserving operations, (ii)~descends from Piani's super-additivity--like inequality (Lemma~\ref{Piani_inequality_lemma}), and (iii)~holds by~\eqref{PPT_bound_thm_proof_eq0}.

Now, let us specialise the above inequality to a PPT state $\rho_{AB}$. In this case we have that
\bb
D_\PPT\left(\rho_{AB} \otimes \tau_{CD}\right) &= \inf_{\sigma_{AC:BD}\in \PPT^1} D\left(\rho_{AB} \otimes \tau_{CD}\, \Big\|\, \sigma_{AC:BD} \right) \\
&\textleq{(iv)} \inf_{\sigma'_{CD}\in \PPT^1} D\left(\rho_{AB} \otimes \tau_{CD}\, \Big\|\, \rho_{AB} \otimes \sigma'_{CD} \right) \\
&= D_\PPT (\tau_{CD})\, .
\label{PPT_bound_thm_proof_eq2}
\ee
In~(iv), for an arbitrary $\sigma'_{CD}\in \PPT^1$ we chose as an ansatz $\sigma_{AC:BD} = \rho_{AB} \otimes \sigma'_{CD}\in \PPT^1$, where we made use of the fact that $\rho_{AB}$ itself is PPT.

Together, \eqref{PPT_bound_thm_proof_eq1}~and~\eqref{PPT_bound_thm_proof_eq2} imply that $D_\PPT^\ppt\left(\omega_{A'B'}\right)\leq 0$, which entails that $D_\PPT^\ppt\left(\omega_{A'B'}\right) = 0$ because $D_\PPT^\ppt$ cannot be negative. Since $D_\PPT^\ppt$ is faithful (see the discussion around~\eqref{faithfulness_DKK}), we conclude that $\omega_{A'B'}$ is a PPT state, concluding the proof of claim~(a).

Claim~(b) follows immediately from the fact that an ebit is at trace norm distance $1/2$ from the set of PPT states. Combining this with Corollary~\ref{sq_bound_cor} yields~(c), which directly implies~(d).
\end{proof}

Some observations are in order:
\begin{enumerate}

\item Using Lemma~\ref{general_bounds_lemma} at full blast and combining it with Lemmata~\ref{elementary_properties_DKK_lemma}--\ref{strong_superadditivity_DKK_lemma} actually yields the more general bound
\bb
E_{d,\, \locc^{cc}}(\rho_{AB}) \leq \liminf_{n\to\infty} \frac1n \sup_{\tau_{CD}} \left\{ D_\K^\mk\left(\rho_{AB}^{\otimes n} \otimes \tau_{CD}\right) - D_\K^\mk\left(\tau_{CD}\right) \right\} ,
\ee
valid for all pairs $(\K, \mk)$ as in~\eqref{pairs}. As usual, $CD$ is assumed to be finite-dimensional. However, in the case of NPT $\rho_{AB}$ we do not know how to upper bound the above quantity in any useful way with a function of $\rho_{AB}$.

\item The above result has a much stronger analogue in the resource theory of coherence: there, one can prove that correlated catalysts cannot change neither the distillable coherence nor the coherence cost under incoherent operations (IO). We discuss this in detail in Section~\ref{sec:coherence}.

\item Until now, all results that we know of prove that (in certain cases) correlated catalysts cannot improve asymptotic rates. At this point the opposite question becomes pressing: \emph{do we know of a single case where correlated catalysts manage to improve an asymptotic rate?}
\end{enumerate}

\tcb{
\subsection{No `perpetuum mobile' of entanglement under correlated catalytic transformations}

As explained in the main text, the inequality
\bb
E_{d,\, \FF^{cc}}(\rho_{AB}) \leq E_{c,\,\FF^{cc}}(\rho_{AB})
\label{no_perpetuum_mobile_SM}
\ee
expresses the impossibility of running a process that using operations in $\FF$ assisted by correlated catalysts creates unlimited entanglement from a finite supply of states $\rho_{AB}$ by distilling entanglement out of them and re-creating the same states via entanglement dilution --- a sort of perpetual motion machine for entanglement. It is only reasonable that we expect~\eqref{no_perpetuum_mobile_SM} to hold for any logically consistent class of free operations $\FF$. And yet, the inequality in~\eqref{no_perpetuum_mobile_SM} does not seem to follow from the results in Proposition~\ref{bounds_NE''_prop}, and in particular from~\eqref{bound_distillable_NE''} and~\eqref{bound_cost_NE''}, simply because $D_\PPT^\infty(\rho_{AB})\geq D_\PPT^{\mk,\infty}(\rho_{AB})$ holds in general due to the data processing inequality~\eqref{data_processing_DK_DKK}, while we would need the opposite inequality to chain~\eqref{bound_distillable_NE''} and~\eqref{bound_cost_NE''}. Granted, we are not aware of any explicit example of a state showing that, for instance, $D_\PPT^\infty > D_\PPT^{\ppt}$ can happen in general. However, based on the results of~\cite{DNE-distillable} we believe that such states exist. If that is the case, then Proposition~\ref{bounds_NE''_prop} cannot be used to establish~\eqref{no_perpetuum_mobile_SM} for operations in $\PPTP^{cc}$.

However, here we will show that these hurdles are just technicalities stemming from the the peculiarities of the set $\PPTP^{cc}$. For more well-behaved sets of free operations like $\PPT$ or $\locc$ there is an independent argument recovering~\eqref{no_perpetuum_mobile_SM}, and thus establishing the logical consistency of entanglement manipulation assisted by correlated catalysts, for almost all states. More precisly, our proof works for every state when $\FF=\PPT$, and for all states except possibly NPT bound entangled states (provided they exist~\cite{Horodecki-open-problems, OpenProblem-NPT-bound}) when $\FF=\locc$. 
\tcb{This was mentioned without proof in~\cite{kondra_2023} for the case of LOCC, and below we record a complete derivation that also considers PPT operations.}

\begin{prop} \label{no_perpetuum_mobile_prop_SM}
For all bipartite states $\rho_{AB}$ except possibly NPT bound entangled states, it holds that
\bb
E_{d,\, \locc^{cc}}(\rho_{AB}) \leq E_{c,\,\locc^{cc}}(\rho_{AB})\, .
\label{no_perpetuum_mobile_LOCC}
\ee
For all bipartite states $\rho_{AB}$,
\bb
E_{d,\, \PPT^{cc}}(\rho_{AB}) \leq E_{c,\,\PPT^{cc}}(\rho_{AB})\, .
\label{no_perpetuum_mobile_PPT}
\ee
\end{prop}

Before we come to the proof of Proposition~\ref{no_perpetuum_mobile_prop_SM}, we need to recall the notion of \emph{maximal transformation rate}, introduced in~\cite[Definition~11]{ferrari_2022} and employed extensively in~\cite{concurrent, Shiraishi2023, kondra_2023}, where it is referred to as \emph{marginal} transformation rate. For two states $\rho$ (input) and $\sigma$ (output) and a class of free operations $\FF$, this is defined by as the supremum of all maximally achievable rates $\rho\to\sigma$. Here, a number $R>0$ is called a maximally achievable rate $\rho\to\sigma$ if for all $\e>0$ and for all sufficiently large $n$ there exists a transformation $\Lambda_n\in \FF$, going from $n$ copies of the input system to $m=\ceil{Rn}$ copies of the output system, with the property that the $j^{\text{th}}$-system reduction of $\omega^{(m)} \coloneqq \Lambda_n\big(\rho^{\otimes n}\big)$ resembles $\sigma$ up to an error $\e$, for all $j=1,\ldots, m$. In formula, this can be expressed as $\max_j \frac12 \left\|\,\omega^{(m)}_j - \sigma \,\right\|_1\leq \e$, where $\omega^{(m)}_j$ denotes the reduced state of $\omega^{(m)}$ on the $j^{\text{th}}$ system. Formalizing what we have just explained in words, the maximal transformation rate can thus be defined by
\bb
R_{\FF_m}(\rho \to \sigma) \coloneqq \sup\left\{ R :\ 
\rho^{\otimes n} \overset{\FF}{\longrightarrow} \omega^{(\ceil{Rn})},\ \lim_{n\to\infty} \max_{j=1,\ldots, \ceil{Rn}} \frac12 \left\|\,\omega^{(\ceil{Rn})}_j - \sigma\, \right\|_1 = 0 \right\} .
\label{max_rate}
\ee
There is also a version of the above rate that includes assistance by correlated catalysts, and can be defined formally by
\bb
R_{\FF_m^{cc}}(\rho \to \sigma) \coloneqq \sup\left\{ R :\ 
\rho^{\otimes n} \overset{\FF^{cc}}{\longrightarrow} \omega^{(\ceil{Rn})},\ \lim_{n\to\infty} \max_{j=1,\ldots, \ceil{Rn}} \frac12 \left\|\,\omega^{(\ceil{Rn})}_j - \sigma\, \right\|_1 = 0 \right\} .
\label{max_rate_correlated_catalysts}
\ee

\begin{note}
In~\cite{concurrent}, the above rate $R_{\FF_m^{cc}}$ is denoted by $R_{\mathrm{mc}}$.
\end{note}

Clearly,
\bb
R_{\FF_m}(\rho \to \sigma) \leq R_{\FF_m^{cc}}(\rho \to \sigma)
\label{max_rate_upper_bounded_max_cc_rate}
\ee
for all pairs of states $\rho,\sigma$ and all classes of free operations $\FF$. Remarkably, it has been shown recently~\cite[Proposition~5]{concurrent} that provided that $\FF$ is closed under both sequential and parallel composition, the inequality in~\eqref{max_rate_upper_bounded_max_cc_rate} is actually an equality for all $\FF$-distillable states, i.e.
\bb
E_{d,\, \FF}(\rho) > 0 \quad \Longrightarrow\quad R_{\FF_m}(\rho\to\sigma) = R_{\FF_m^{cc}}(\rho\to\sigma)\, .
\label{m=mc_for_distillable_states}
\ee
The proof of~\cite[Proposition~5]{concurrent} is presented for the case of $\FF = \locc$, but upon inspection one can see that it works equally well for all sets of free operations $\FF$ that are closed under both sequential and parallel composition. These requirements mean that for all $\Lambda_1,\Lambda_2\in \FF$ it should hold that $\Lambda_1\circ \Lambda_2\in \FF$ (whenever this sequential composition makes sense) as well as $\Lambda_1\otimes \Lambda_2\in \FF$.

In the case of entanglement theory, the results of~\cite[Proposition~7]{concurrent} also imply that for any class $\FF$ that is closed under both sequential and parallel composition and contains the set of one-way LOCCs, say from Alice to Bob, all $\FF$-distillable state $\rho=\rho_{AB}$ and all pure states $\phi=\phi_{AB}$ satisfy that~\cite[Proposition~7]{concurrent}
\bb
R_{\FF_m}(\rho\rightarrow\phi) = R_{\FF^{cc}}(\rho\rightarrow\phi) = R_{\FF_m^{cc}}(\rho\rightarrow\phi) = R_\FF(\rho\rightarrow\phi) = \frac{E_{d,\,\FF}(\rho)}{S(\phi_A)}\, .
\label{pure_collapse_rates}
\ee
Again, the proof of~\cite[Proposition~7]{concurrent} is stated for LOCC transformations only, but it works without significant changes for all other $\FF$ satisfying the above assumptions. The fact that $\FF$ includes one-way LOCCs is needed only to deduce the continuity of the $\FF$-distillable entanglement on pure states via the hashing bound~\cite{devetak2005} --- cf.~\cite[Eq.~(36)]{concurrent}.

Another remarkable property of maximal transformation rates is the following `transitivity property', which generalizes~\cite[Lemma~2]{kondra_2023}. 

\begin{prop} \label{prop:MarginalRates}
Let $\FF$ be any set of free operations that is closed under both sequential and parallel composition. Then the maximal transformation rates fulfill the relation 
\begin{equation}
R_{\FF_m}(\rho\rightarrow\omega)\cdot R_{\FF_m}(\omega\rightarrow\sigma)\leq R_{\FF_m}(\rho\rightarrow\sigma)
\end{equation}
for any three states $\rho$, $\omega$, and $\sigma$. 
\end{prop}

\begin{proof}
We can generalize the same type of `chessboard' argument found in~\cite[Lemma~2]{kondra_2023}. To appreciate where the chessboard structure comes from, we point the reader to~\cite[Figure~1]{kondra_2023}. Call $r_1\coloneqq R_{\FF_m}(\rho\rightarrow\omega)$ and $r_2\coloneqq R_{\FF_m}(\omega\rightarrow\sigma)$. For any given $\e>0$, let $k\in \N_+$ be large enough so that there exists $\Theta\in \FF$ with the property that 
\bb
\max_{j=1,\ldots, \ceil{r_2 k}} \frac12 \left\|\,\Theta\big(\omega^{\otimes k}\big)_j - \sigma\, \right\|_1\leq \frac{\e}{2}\, ,
\label{column_transformation_error}
\ee
where as before $\Theta\big(\omega^{\otimes k}\big)_j$ denotes the reduced state of $\Theta\big(\omega^{\otimes k}\big)$ on the $j^\text{th}$ subsystem. Such a $k$ evidently exists due to the definition of $R_{\FF_m}(\omega\to\sigma)$. For the same reason, for all large enough $m$ we can also find $\Lambda_m$ such that $\xi^{(\ceil{r_1 m})} \coloneqq \Lambda_m\big(\rho^{\otimes m}\big)$ satisfies that
\bb
\max_{\ell=1,\ldots, \ceil{r_1 m}} \frac12 \left\|\,\xi^{(\ceil{r_1 m})}_\ell - \omega\, \right\|_1\leq \frac{\e}{2k}\, .
\ee

Now, suppose that we start with $n$ copies of $\rho$, with $n$ asymptotically large. Up to discarding at most $k-1$ of them, i.e.\ a constant number as $n\to\infty$, we can arrange them into $k$ `rows' of $m = \floor{n/k}$ copies each, thus forming a rectangular, $k\times m$ array. We now apply (i)~$\Lambda_m$ to every row, and after that (ii)~$\Theta$ to every column. Note that step~(i) transforms the array of quantum systems into a $k\times \ceil{r_1 m}$ array, and this gets in turn transformed into a $\ceil{r_2 k}\times \ceil{r_1 m}$ array upon completion of step~(ii). The state obtained after~(i) is simply
\bb
\Big(\xi^{(\ceil{r_1 m})}\Big)^{\otimes k} = \Big(\Lambda_{m}\big(\rho^{\otimes m}\big)\Big)^{\otimes k}\, .
\ee
Since for every $\ell = 1,\ldots, \ceil{r_1 m}$ it holds that $\frac12 \left\|\,\xi^{(\ceil{r_1 m})}_\ell - \omega\, \right\|_1\leq \frac{\e}{2k}$, we also have that
\bb
\max_{\ell=1,\ldots, \ceil{r_1 m}} \frac12 \left\|\big(\xi^{(\ceil{r_1 m})}_\ell\big)^{\otimes k} - \omega^{\otimes k} \right\|_1\leq \frac{\e}{2}\, ,
\label{column_error_estimate}
\ee
by a simple iterated application of the triangle inequality. Consider now the final output state obtained after step~(ii), denoted by $\zeta^{(\ceil{r_2 k} \ceil{r_1 m})} = \Theta^{\otimes \ceil{r_1 m}} \Big( \big( \xi^{(\ceil{r_1 m})}\big)^{\otimes k} \Big)$. Each one of the 
\bb
\ceil{r_2 k} \cdot \ceil{r_1 m} \geq r_1 r_2\, k m = r_1 r_2\, k \floor{n/k} \geq r_1 r_2 (n-k)
\ee
output systems corresponds to a definite row and column of the rectangular array. Let $j,\ell$ denote the row and column indices, respectively. By tracing away all columns but the $\ell^{\text{th}}$ one, we have that
\bb
\zeta^{(\ceil{r_2 k} \ceil{r_1 m})}_{j,\ell} = \Theta\Big(\big( \xi^{(\ceil{r_1 m})}_\ell\big)^{\otimes k}\Big)_j\, ,
\label{eliminate_all_but_one_column}
\ee
implying that
\bb
\frac12 \left\| \,\zeta^{(\ceil{r_2 k} \ceil{r_1 m})}_{j,\ell} - \sigma\, \right\|_1 &\textleq{(i)} \frac12 \left\| \,\zeta^{(\ceil{r_2 k} \ceil{r_1 m})}_{j,\ell} - \Theta\big(\omega^{\otimes k} \big)_j \, \right\|_1 + \frac12 \left\| \,\Theta\big(\omega^{\otimes k} \big)_j - \sigma \, \right\|_1 \\
&\texteq{(ii)} \frac12 \left\| \,\Theta\Big(\big( \xi^{(\ceil{r_1 m})}_\ell\big)^{\otimes k} - \omega^{\otimes k} \Big)_j\, \right\|_1 + \frac12 \left\| \,\Theta\big(\omega^{\otimes k} \big)_j - \sigma \, \right\|_1 \\
&\textleq{(iii)} \frac12 \left\| \,\big( \xi^{(\ceil{r_1 m})}_\ell\big)^{\otimes k} - \omega^{\otimes k}\, \right\|_1 + \frac12 \left\| \,\Theta\big(\omega^{\otimes k} \big)_j - \sigma \, \right\|_1 \\
&\textleq{(iv)} \e\, .
\ee
Here, (i)~holds due to the triangle inequality, (ii)~because of~\eqref{eliminate_all_but_one_column}, (iii)~follows from the data processing inequality for the trace norm, and finally (iv)~can be verified by combining~\eqref{column_transformation_error} and~\eqref{column_error_estimate}. 

Since this works for every $j,\ell$, every subsystem of $\zeta^{(\ceil{r_2 k} \ceil{r_1 m})}$ is approximately in the state $\sigma$, up to an error $\e$. Since there are at least $r_1r_2 (n-k)$ such sub-systems and $\e>0$ was arbitrary, this shows that the rate $\lim_{n\to\infty} \frac{r_1r_2 (n-k)}{n} = r_1r_2 = R_{\FF_m}(\rho\rightarrow\omega)\, R_{\FF_m}(\omega\rightarrow\sigma)$ is maximally achievable for $\rho\to\sigma$, thus completing the proof.
\end{proof}

We are now ready to present a complete proof of Proposition~\ref{no_perpetuum_mobile_prop_SM}.

\begin{proof}[Proof of Proposition~\ref{no_perpetuum_mobile_prop_SM}]
We can write the following chain of inequalities, valid for any $\FF$-distillable state $\rho$:
\bb
\frac{E_{d,\, \FF^{cc}}(\rho)}{E_{c,\, \FF^{cc}}(\rho)} &= R_{\FF^{cc}}(\Phi_2\to\rho)\, R_{\FF^{cc}}(\rho\to \Phi_2) \\
&\textleq{(i)} R_{\FF_m^{cc}}(\Phi_2\to\rho)\, R_{\FF_m^{cc}}(\rho\to \Phi_2) \\
&\texteq{(ii)} R_{\FF_m}(\Phi_2\to\rho)\, R_{\FF_m}(\rho \to \Phi_2) \\
&\textleq{(iii)} R_{\FF_m}(\Phi_2 \to\Phi_2) \\
&\textleq{(iv)} 1\, .
\label{chain_inequalities_no_perpetuum_mobile}
\ee
Here, (i)~holds because in general $R_{\FF^{cc}}\leq R_{\FF_m^{cc}}$ as per~\eqref{max_rate_upper_bounded_max_cc_rate}, (ii)~follows from~\eqref{m=mc_for_distillable_states}, (iii)~is an application of Proposition~\ref{prop:MarginalRates}, and finally~(iv) is due to~\eqref{pure_collapse_rates}.

This proves~\eqref{no_perpetuum_mobile_LOCC} for the case of LOCC-distillable states $\rho$. When $\rho$ is PPT, the left-hand side of~\eqref{no_perpetuum_mobile_LOCC} vanishes due to our main result, Theorem~\ref{PPT_bound_thm}, and there is nothing to prove. As for~\eqref{no_perpetuum_mobile_PPT}, the situation is even simpler here: for any given $\rho$, either $\rho$ is PPT, and then our main result tells us once again that the left-hand side of~\eqref{no_perpetuum_mobile_PPT} vanishes, with nothing left to prove, or else $\rho$ is NPT, and in this case we know that it must be PPT-distillable~\cite{Eggeling2001}, entailing that we can use~\eqref{chain_inequalities_no_perpetuum_mobile} to establish~\eqref{no_perpetuum_mobile_PPT}.
\end{proof}
}

\subsection{Applications to the resource theory of coherence}\label{sec:coherence}

Throughout this section, we will look into the usefulness of catalysts for the resource theory of coherence~\cite{aberg_2006, baumgratz_2014, winter_2016, streltsov_2017}. In this resource theory, each allowed system with finite-dimensional Hilbert space $\HH\simeq \C^d$ is equipped with a privileged orthonormal basis $\{\ket{i}\}_{i=1,\ldots d}$, conventionally identified with the computational basis. Possible choices of free operations include strictly incoherent operations (SIO)~\cite{winter_2016, yadin_2016, zhao_2018, lami_2019-1, lami_2020-1}, incoherent operations (IO)~\cite{baumgratz_2014, winter_2016}, dephasing-covariant incoherent operations (DIO)~\cite{chitambar_2016, chitambar_2016-1, marvian_2016, chitambar_2018}, maximal incoherent operations (MIO)~\cite{aberg_2006, regula_2017}, and others~\cite{chitambar_2016, vicente_2017}.

Depending on the set of free operations one chooses, the asymptotic behaviour of the theory differs. For example, bound coherence under SIO is generic~\cite{zhao_2018, lami_2019-1, lami_2020-1}. While there is no bound coherence at all under IO, i.e., all states that are not incoherent are IO distillable, the corresponding resource theory still admits generic irreversibility~\cite{winter_2016}. Finally, the resource theory of coherence is known to be reversible under MIO and even under the strictly smaller class of DIO~\cite{chitambar_2018,berta_2022}. In light of this discussion, it is  natural to ask ourselves whether the use of catalysts could improve either distillation or dilution of coherence under IO or SIO. 
It is somewhat remarkable that we can answer this question in the negative for dilution under both SIO and IO, and for distillation under IO, using the techniques we have developed here together with some off-the-shelf results. The only remaining case that we leave open here is that of SIO distillation, which may or may not be affected asymptotically by the introduction of catalysts --- but there is anyway no hope of obtaining a theory which is asymptotically reversible, since SIO constitute a subset of IO.

Let us start by recalling some definitions. A quantum channel $\Lambda:\MM\big(\C^d\big) \to \MM\big(\C^{d'}\big)$ acting on the set of $d\times d$ matrices $\MM\big(\C^d\big)$ is said to be an \emph{incoherent operation} (IO) if it admits a Kraus representation $\Lambda(\cdot)=\sum_\alpha K_\alpha (\cdot) K_\alpha^\dag$ with the property that for all $\alpha$ and $i\in \{1,\ldots, d\}$ there exists $j\in \{1,\ldots, d'\}$ such that $K_\alpha\ket{i} \propto \ket{j}$. If at the same time it also holds that for all $\alpha$ and $i\in \{1,\ldots, d'\}$ there exists $j\in \{1,\ldots, d\}$ such that $K_\alpha^\dag \ket{i} \propto \ket{j}$, then we say that $\Lambda$ is a \emph{strictly incoherent operation} (SIO). With this terminology in place, we can now give the following definitions.

Given two finite-dimensional systems $A,B$, two states $\rho_A$ on $A$ and $\omega_B$ on $B$, and a class of operations $\FF\in \left\{\SIO,\IO\right\}$, we can re-adapt Definition~\ref{catalysts_def} by writing:
\begin{enumerate}[(i')]
\item $\rho_A \overset{\FF^{c}}{\longrightarrow} \omega_B$, if there exists a finite-dimensional state $\tau_C$ and an operation $\Lambda\in \FF\left(AC\to BC\right)$ such that
\bb
\Lambda\left(\rho_A \otimes \tau_C \right) = \omega_B \otimes \tau_C\, ;
\ee
\item $\rho_A \overset{\FF^{cc}}{\longrightarrow} \omega_B$, if there exists a finite-dimensional state $\tau_C$ and some $\Lambda\in \FF\left(AC\to BC\right)$ such that
\bb
\Tr_C \Lambda\left(\rho_A \otimes \tau_C \right) = \omega_B\, ,\qquad \Tr_B \Lambda\left(\rho_A \otimes \tau_C\right) = \tau_C\, ;
\ee
\item $\rho_A \overset{\FF^{mc}}{\longrightarrow} \omega_B$, if there exists a finite collection of finite-dimensional states $\tau_{C_j}$, $j=1,\ldots, k$, and an operation $\Lambda\in \FF\left(AC\to BC\right)$, where $C\coloneqq C_1\ldots C_k$, such that
\bb
\Lambda\left(\rho_A \otimes \tau_C \right) = \omega_B\! \otimes \tau'_C\, , \quad \Tr_{\myhat{C}_j} \!\tau'_{C} = \tau_{C_j}\ \ \forall\ j=1,\ldots, k\, ,
\ee
where $\tau_C\coloneqq \bigotimes_j \tau_{C_j}$ is the initial state of the overall catalyst, and $\myhat{C}_j$ denotes the system obtained by tracing away all of the $C$ systems except from the $j^\text{th}$ one.
\item $\rho_A \overset{\FF^{mcc}}{\longrightarrow} \omega_B$, if there exists a finite collection of finite-dimensional states $\tau_{C_j}$, $j=1,\ldots, k$, and an operation $\Lambda\in \FF\left(AC\to BC\right)$, where $C\coloneqq C_1\ldots C_k$, such that
\bb
\Tr_C \Lambda\left(\rho_A \otimes \tau_C \right) = \omega_B\, ,\qquad \Tr_{B \myhat{C}_j} \Lambda\left(\rho_A \otimes \tau_C\right) = \tau_{C_j} \quad \forall\ j=1,\ldots, k\, ,
\ee
where $\tau_C\coloneqq \bigotimes_j \tau_{C_j}$, and once more $\myhat{C}_j$ denotes the system obtained by tracing away all $C$ systems except from the $j^\text{th}$ one.
\end{enumerate}
For $\widetilde{\FF}\in \left\{\FF, \FF^c, \FF^{cc}, \FF^{mc}, \FF^{mcc}\right\}$ and $\e\in [0,1]$, we write $\rho_A \overset{\widetilde{\FF}}{\longrightarrow}\, \approx_\e\! \omega_B$ if there exists a state $\omega'_B$ such that
\bb
\rho_A \overset{\widetilde{\FF}}{\longrightarrow} \omega'_B\, ,\qquad \frac12\left\| \omega'_B - \omega_B \right\|_1\leq \e\, .
\ee
For every pair of states $\rho_A$ and $\omega_B$, the corresponding \emph{asymptotic rate} is given by
\bb
R_{\widetilde{\FF}}\left( \rho_A \to \omega_B \right) \coloneqq \sup\left\{ R:\ \rho_A^{\otimes n} \overset{\widetilde{\FF}}{\longrightarrow}\, \approx_{\e_n}\! \omega_B^{\otimes \ceil{Rn}},\ \lim_{n\to\infty} \e_n = 0 \right\} .
\ee
The \emph{distillable coherence} and the \emph{coherence cost under operations in $\widetilde{\FF}$} are defined by
\bb
C_{d,\, \widetilde{\FF}}\left(\rho\right) \coloneqq R_{\widetilde{\FF}}\left( \rho \to \ketbra{+} \right) ,\qquad C_{c,\, \widetilde{\FF}}\left(\rho\right) \coloneqq \frac{1}{R_{\widetilde{\FF}}\left( \ketbra{+} \to \rho\right)}\, ,
\ee
where
\bb
\ket{+} \coloneqq \frac{1}{\sqrt2} \left(\ket{0} + \ket{1}\right)
\ee
denotes the single-qubit \emph{coherence bit}.

The distillable coherence and the coherence cost under IO (unassisted by catalysis) have been computed in~\cite{winter_2016}. Therein, it is already remarked that the coherence cost remains the same under SIO. The problem of distillation under SIO was instead solved in~\cite{lami_2020-1}. 

Concerning IO distillation, it holds that~\cite[Theorem~6]{winter_2016}
\bb
C_{d,\IO}(\rho) = C_r(\rho) \coloneqq S\left(\Delta(\rho)\right) - S(\rho)\, ,
\label{relative_entropy_coherence}
\ee
where $S(\omega) \coloneqq -\Tr \omega \log_2\omega$ denotes the \emph{von Neumann entropy}, $\Delta(\cdot)\coloneqq \sum_i \ketbra{i}(\cdot)\ketbra{i}$ is the completely dephasing channel, and $C_r$ is called the \emph{relative entropy of coherence}. 

Distillation is in general much more difficult with SIO, a set of operations under which it is known that almost all states (in a measure-theoretic sense) are undistillable~\cite{lami_2019-1}. 
More precisely, we have that~\cite[Theorem~3]{lami_2020-1}
\bb
C_{d,\SIO}(\rho) = Q(\rho) \coloneqq S\left(\Delta(\rho)\right) - S(\bar{\rho})\, ,
\label{quintessential_coherence}
\ee
where the state $\bar{\rho}$ is defined by
\bb
\bar{\rho}_{ij} \coloneqq \left\{ \begin{array}{ll} \rho_{ij} & \text{ if $|\rho_{ij}| = \sqrt{\rho_{ii}\rho_{jj}}$,} \\[1ex] 0 & \text{ otherwise.} \end{array} \right.
\label{trimmed}
\ee
The quantity $Q$ is called \emph{quintessential coherence}.

At this point it is perhaps surprising that SIO and IO are entirely equivalent when it comes to coherence dilution. Indeed, it holds that~\cite[Theorem~8]{winter_2016}
\bb
C_{c,\SIO}(\rho) = C_{c,\IO}(\rho) = C_f(\rho) \coloneqq \inf_{\rho = \sum_x p_x \psi_x} \sum_x p_x\, S\left( \Delta(\psi_x)\right)\, ,
\label{coherence_formation}
\ee
where the infimum is over all convex decompositions $\rho=\sum_x p_x \psi_x$ of $\rho$ into pure states $\psi_x = \ketbra{\psi_x}$, and $C_f$ is called the \emph{coherence of formation}.

Before we state the main result of this section, it is useful to record some known facts concerning the above monotones.

\begin{lemma} \label{ssa_Cr_Cf_lemma}
The relative entropy of coherence and the coherence of formation, defined by~\eqref{relative_entropy_coherence} and~\eqref{coherence_formation}, respectively, are monotonic under IO, strongly super-additive, completely additive, asymptotically continuous, and normalized. 
\end{lemma}

\begin{proof}
The fact that the relative entropy of coherence and the coherence of formation are completely additive, asymptotically continuous, and normalized IO monotones is established in~\cite{winter_2016}. The strong super-additivity of $C_f$ is shown in~\cite{Liu2018}. As for $C_r$, its strong super-additivity follows from the fact that the same property holds for the coherent information $I_c(A\rangle A')_\omega \coloneqq S(\omega_{A'}) - S(\omega_{AA'}) = - H(A|A')_\omega$. To establish the link between the two concepts, one needs to associate to every $d$-dimensional state $\rho_A$ a corresponding `maximally correlated' state $\widetilde{\rho}_{AA'}\coloneqq \sum_{i,j} \rho_{ij} \ketbraa{i}{j}_A\otimes \ketbraa{i}{j}_{A'}$. Then,
\bb
C_r(\rho_A) = I_c(A\rangle A')_{\widetilde{\rho}}\, .
\ee
Now, if $\rho = \rho_{AB}$ is bipartite, we can use~\cite[Theorem~11.16]{nielsen_2011} to write
\bb
C_r(\rho_{AB}) &= I_c(AB\rangle A'B')_{\widetilde{\rho}} \\
&= - H(AB|A'B')_{\widetilde{\rho}} \\
&\geq - H(A|A')_{\widetilde{\rho}} - H(B|B')_{\widetilde{\rho}} \\
&= I_c(A\rangle A')_{\widetilde{\rho}} + I_c(B\rangle B')_{\widetilde{\rho}} \\
&= C_r(\rho_A) + C_r(\rho_B)\, .
\ee
This concludes the proof.
\end{proof}

We are now ready to establish the following.

\begin{thm} \label{coherence_catalysts_thm}
The distillable coherence and the coherence cost under IO, as well as the coherence cost under SIO, do not change if one allows assistance by either catalysts, or correlated catalysts, or marginal catalysts, or marginal correlated catalysts. In formulae, for all states $\rho$ it holds that
\begin{align}
C_{d,\IO}(\rho) &= C_{d,\IO^c}(\rho) = C_{d,\IO^{cc}}(\rho) = C_{d,\IO^{mc}}(\rho) = C_{d,\IO^{mcc}}(\rho) = C_r(\rho)\, , \label{IO_distillation_catalysts} \\[1ex]
C_{c,\IO}(\rho) &= C_{c,\IO^c}(\rho) = C_{c,\IO^{cc}}(\rho) = C_{c,\IO^{mc}}(\rho) = C_{c,\IO^{mcc}}(\rho) \label{IO_dilution_catalysts}  \\
&= C_{c,\SIO}(\rho) = C_{c,\SIO^c}(\rho) = C_{c,\SIO^{cc}}(\rho) = C_{c,\SIO^{mc}}(\rho) = C_{c,\SIO^{mcc}}(\rho) = C_f(\rho)\, , \label{SIO_dilution_catalysts}
\end{align}
where the relative entropy of coherence $C_r$ and the coherence of formation $C_f$ are defined respectively by~\eqref{relative_entropy_coherence} and~\eqref{coherence_formation}.
\end{thm}

\begin{proof}
The a priori largest quantity in~\eqref{IO_distillation_catalysts} is $C_{d,\IO^{mcc}}(\rho)$, so it suffices to show that $C_{d,\IO^{mcc}}(\rho) \leq C_r(\rho)$. This is a consequence of the results that we established in Section~\ref{intro_sec} for entanglement theory, which can be immediately noticed to hold in the same way for the theory of coherence. In particular, the desired statement follows from~\eqref{chain} in Lemma~\ref{general_bounds_lemma}, applied to $M=C_r$ thanks to Lemma~\ref{ssa_Cr_Cf_lemma}. Analogously, the smallest quantity in~\eqref{IO_dilution_catalysts}--\eqref{SIO_dilution_catalysts} is $C_{c,\IO^{mcc}}(\rho)$, so it suffices to prove that $C_{c,\IO^{mcc}}(\rho)\geq C_f(\rho)$. This comes once again from~\eqref{chain} in Lemma~\ref{general_bounds_lemma}, applied to $M=C_f$ thanks to Lemma~\ref{ssa_Cr_Cf_lemma}.
\end{proof}

\begin{rem}
Since coherence is already reversible under MIO and even DIO, with 
\bb
C_{d,\MIO}(\rho) = C_{d,\DIO}(\rho) = C_{c,\MIO}(\rho) = C_{d,\DIO}(\rho) = C_r(\rho)\, ,
\ee
there is no reason to expect that catalysts could help here. In fact, if they did, the theory would trivialise, because in a non-trivial theory the distillable resource can never exceed the resource cost. And indeed, since $C_r$ is strongly super-additive, completely additive, asymptotically continuous, and normalized, inequality~\eqref{chain} tells us that catalysts, in any form, cannot increase the DIO/MIO distillable coherence above $C_r$, nor decrease the DIO/MIO cost below $C_r$.
\end{rem}

\begin{rem}
What prevents us from extending Theorem~\ref{coherence_catalysts_thm} to the case of coherence distillation under SIO assisted by catalysts is that the SIO distillable coherence, i.e., the quintessential coherence~\eqref{quintessential_coherence}, despite being strongly super-additive, additive, and normalized, is only \emph{upper} instead of \emph{lower} semi-continuous. We therefore leave open the question of whether either of the quantities $C_{d,\SIO^{c}}$, $C_{d,\SIO^{cc}}$, $C_{d,\SIO^{mc}}$, $C_{d,\SIO^{mcc}}$ can be strictly larger than $C_{d,\SIO}=Q$ at least in some cases.
\end{rem}

\section*{Appendix}

\subsection{Proof of Lemma~\ref{general_bounds_lemma}} \label{general_bounds_proof_app}

\begin{manuallemma}{\ref{general_bounds_lemma}}
Let $M$ be a strongly super-additive $\FF$-monotone. If $M$ is also either quasi-normalized and asymptotically continuous, or normalized and lower semi-continuous, then
\begin{align}
E_{d,\,\FF^{cc}}(\rho_{AB}) &\leq \liminf_{n\to\infty} \frac1n \sup_{\tau_{CD}} \left\{ M\left(\rho_{AB}^{\otimes n} \otimes \tau_{CD} \right) - M(\tau_{CD}) \right\} , \tag{\ref{bound_distillable_ssa}} \\
E_{d,\,\FF^{mcc}}(\rho_{AB}) &\leq \liminf_{n\to\infty} \frac1n \sup_{\tau_{CD} = \bigotimes_j \tau_{C_j D_j}} \left\{ M\left(\rho_{AB}^{\otimes n} \otimes \tau_{CD} \right) - \sum_j M(\tau_{C_jD_j}) \right\} , \tag{\ref{bound_distillable_mcc_ssa}} 
\end{align}
where all the suprema are over finite-dimensional catalysts $\tau_{CD}$. If $M$ is strongly super-additive and asymptotically continuous, then
\begin{align}
E_{c,\,\FF^{cc}}(\rho_{AB}) &\geq \frac{M^\infty(\rho_{AB})}{\limsup_{m\to\infty} \frac1m \sup_{\tau_{CD}} \left\{ M\left(\Phi_2^{\otimes m} \otimes \tau_{CD} \right) - M(\tau_{CD}) \right\}}\, , \tag{\ref{bound_cost_ssa}} \\
E_{c,\,\FF^{mcc}}(\rho_{AB}) &\geq \frac{M^\infty(\rho_{AB})}{\limsup_{m\to\infty} \frac1m \sup_{\tau_{CD} = \bigotimes_j \tau_{C_j D_j}} \left\{ M\left(\Phi_2^{\otimes m} \otimes \tau_{CD} \right) - \sum_j M(\tau_{C_j D_j}) \right\}}\, . \tag{\ref{bound_cost_mcc_ssa}}
\end{align}
where $M^\infty(\rho) \coloneqq \lim_{n\to\infty}\frac1n M\big(\rho^{\otimes n}\big)$ (and the limit exists). If $M$ is strongly super-additive and only lower semi-continuous, then the weaker bounds
\begin{align}
E_{c,\,\FF^{cc}}(\rho_{AB}) &\geq \frac{M(\rho_{AB})}{\limsup_{m\to\infty} \frac1m \sup_{\tau_{CD}} \left\{ M\left(\Phi_2^{\otimes m} \otimes \tau_{CD} \right) - M(\tau_{CD}) \right\}}\, , \tag{\ref{bound_cost_ssa_unregularized}} \\
E_{c,\,\FF^{mcc}}(\rho_{AB}) &\geq \frac{M(\rho_{AB})}{\limsup_{m\to\infty} \frac1m \sup_{\tau_{CD} = \bigotimes_j \tau_{C_j D_j}} \left\{ M\left(\Phi_2^{\otimes m} \otimes \tau_{CD} \right) - \sum_j M(\tau_{C_j D_j}) \right\}} \tag{\ref{bound_cost_mcc_ssa_unregularized}}
\end{align}
hold. Finally, if $M$ is strongly super-additive, completely additive, normalized, and lower semi-continuous, then
\begin{equation}
E_{d,\,\FF^{c}}(\rho_{AB}) \leq E_{d,\,\FF^{cc}}(\rho_{AB}) \leq E_{d,\,\FF^{mcc}}(\rho_{AB}) \leq M(\rho_{AB}) \leq E_{c,\,\FF^{mcc}}(\rho_{AB}) \leq E_{c,\,\FF^{cc}}(\rho_{AB}) \leq E_{c,\,\FF^{c}}(\rho_{AB})\, .
\tag{\ref{chain}}
\end{equation}
\end{manuallemma}

\begin{proof}
Let us start by proving~\eqref{bound_distillable_ssa}, under the assumption that $M$ is strongly super-additive, quasi-normalized, and asymptotically continuous. 
Let $R$ be an achievable rate for entanglement distillation under operations in $\FF^{cc}$. 
Since $R$ is achievable, we can find a sequence of finite-dimensional catalysts $\tau_n = (\tau_n)_{C_nD_n}$ and a sequence of operations $\Lambda_n \in \FF\left(A^n C_n : B^n D_n \to A_0^{\ceil{Rn}} C_n : B_0^{\ceil{Rn}} D_n \right)$, where $A_0$ and $B_0$ are single-qubit systems, such that
\begin{align}
\e_n\coloneqq&\ \frac12 \left\|\Tr_{C_n D_n} \Lambda_n\left(\rho_{AB}^{\otimes n} \otimes \tau_n\right) - \Phi_2^{\otimes \ceil{Rn}} \right\|_1 \tendsn{} 0\, , \label{general_bounds_eq1} \\[1ex]
\tau_{n} =&\ \Tr_{A_0^{\ceil{Rn}} B_0^{\ceil{Rn}}}\! \left[ \Lambda_n\left(\rho_{AB}^{\otimes n} \otimes \tau_n\right) \right] \qquad \forall\ n\, . \label{general_bounds_eq2}
\end{align}
Then
\bb
&\sup_{\tau_{CD}} \left\{ M\left(\rho_{AB}^{\otimes n} \otimes \tau_{CD} \right) - M(\tau_{CD}) \right\} \\
&\qquad \geq M\left(\rho_{AB}^{\otimes n} \otimes \tau_n \right) - M(\tau_n) \\
&\qquad \textgeq{(i)} M\left(\Lambda_n\left(\rho_{AB}^{\otimes n} \otimes \tau_n\right)\right) - M(\tau_n) \\
&\qquad \textgeq{(ii)} M\left(\Tr_{C_n D_n} \Lambda_n\left(\rho_{AB}^{\otimes n} \otimes \tau_n\right)\right) + M\left(\Tr_{A_0^{\ceil{Rn}}B_0^{\ceil{Rn}}} \Lambda_n\left(\rho_{AB}^{\otimes n} \otimes \tau_n\right)\right) - M(\tau_n) \\
&\qquad \texteq{(iii)} M\left(\Tr_{C_n D_n} \Lambda_n\left(\rho_{AB}^{\otimes n} \otimes \tau_n\right)\right) \\
&\qquad \textgeq{(iv)} M \left(\Phi_2^{\otimes \ceil{Rn}}\right) - 2 f(\e_n) \ceil{Rn} - g(\e_n) \\
&\qquad \textgeq{(v)} \ceil{Rn} + o(n) - 2 f(\e_n) \ceil{Rn} - g(\e_n)\, .
\label{general_bounds_eq3}
\ee
Here, (i)~comes from $\FF$-monotonicity, (ii)~from strong super-additivity, (iii)~from~\eqref{general_bounds_eq2}, (iv)~is a consequence of asymptotic continuity, and finally~(v) is just quasi-normalization. Now that we have established the above chain of inequalities, the claim follows by dividing by $n$ and taking first the limit --- more precisely, the $\liminf$ --- as $n\to\infty$, and then the supremum over all achievable rates $R$.

If $M$ is instead strongly super-additive, normalized, and lower semi-continuous, we need to modify the inequality~(iv) in~\eqref{general_bounds_eq3}. We can do so by resorting to an idea first proposed in~\cite{lami_2020-1} and subsequently employed also in~\cite{ferrari_2022}. Let us write
\bb
\sup_{\tau_{CD}} \left\{ M\left(\rho_{AB}^{\otimes n} \otimes \tau_{CD} \right) - M(\tau_{CD}) \right\} &\geq M\left(\Tr_{C_n D_n} \Lambda_n\left(\rho_{AB}^{\otimes n} \otimes \tau_n\right)\right) \\
&\textgeq{(vi)} \sum_{j=1}^{\ceil{Rn}} M\left(\left(\Tr_{C_n D_n} \Lambda_n\left(\rho_{AB}^{\otimes n} \otimes \tau_n\right)\right)_j\right) \\
&\textgeq{(vii)} \ceil{Rn} M\left(\left(\Tr_{C_n D_n} \Lambda_n\left(\rho_{AB}^{\otimes n} \otimes \tau_n\right)\right)_{j_n}\right) ,
\label{general_bounds_eq4}
\ee
where (vi)~follows once again from strong super-additivity, $\left(\Tr_{C_n D_n} \Lambda_n\left(\rho_{AB}^{\otimes n} \otimes \tau_n\right)\right)_j$ denotes the reduced state of $\Tr_{C_n D_n} \Lambda_n\left(\rho_{AB}^{\otimes n} \otimes \tau_n\right)$ to the $j^\text{th}$ pair of qubits ($j=1,\ldots, \ceil{Rn}$), and in~(vii) we selected $j_n\in \{1,\ldots, \ceil{Rn}\}$ such that
\bb
M\left(\left(\Tr_{C_n D_n} \Lambda_n\left(\rho_{AB}^{\otimes n} \otimes \tau_n\right)\right)_{j_n}\right) = \min_{j=1,\ldots, \ceil{Rn}} M\left(\left(\Tr_{C_n D_n} \Lambda_n\left(\rho_{AB}^{\otimes n} \otimes \tau_n\right)\right)_j\right) . 
\ee
Now, diving by $n$ and taking the $\liminf$ as $n\to\infty$ yields
\bb
\liminf_{n\to\infty} \frac1n \sup_{\tau_{CD}} \left\{ M\left(\rho_{AB}^{\otimes n} \otimes \tau_{CD} \right) - M(\tau_{CD}) \right\} \ &\geq\ R \liminf_{n\to\infty} M\left(\left(\Tr_{C_n D_n} \Lambda_n\left(\rho_{AB}^{\otimes n} \otimes \tau_n\right)\right)_{j_n}\right) \\
&\textgeq{(viii)}\ R M(\Phi_2) \\
&\texteq{(ix)}\ R\, .
\ee
Here, (viii)~comes from lower semi-continuity, because
\bb
\frac12 \left\| \left(\Tr_{C_n D_n} \Lambda_n\left(\rho_{AB}^{\otimes n} \otimes \tau_n\right)\right)_{j_n} - \Phi_2\right\|_1 \leq \frac12 \left\| \Tr_{C_n D_n} \Lambda_n\left(\rho_{AB}^{\otimes n} \otimes \tau_n\right) - \Phi_2^{\ceil{Rn}}\right\|_1 = \e_n \tendsn{} 0\, ,
\ee
and (ix)~from normalization. The claim~\eqref{bound_distillable_ssa} follows, once again, upon taking the supremum over all achievable rates $R$. The proof of~\eqref{bound_distillable_mcc_ssa} is entirely analogous.

We therefore move on to~\eqref{bound_cost_ssa}. Start by observing that the limit that defines $M^\infty(\rho)$ exists for all $\rho$ due to Fekete's lemma~\cite{Fekete1923}. Now, let $R$ be an achievable rate for entanglement dilution under operations in $\FF^{cc}$. By achievability, we can find a sequence of finite-dimensional catalysts $\tau_n = (\tau_n)_{C_nD_n}$ and a sequence of operations $\Lambda_n \in \FF\left(A_0^{\floor{Rn}} C_n : B_0^{\floor{Rn}} D_n \to A^n C_n : B^n D_n \right)$, where $A_0$ and $B_0$ are single-qubit systems, such that
\begin{align}
\e_n \coloneqq&\ \frac12 \left\|\Tr_{C_n D_n} \Lambda_n\left(\Phi_2^{\otimes \floor{Rn}} \otimes \tau_n\right) - \rho^{\otimes n} \right\|_1 \tendsn{} 0\, , 
\\[1ex]
\tau_{n} =&\ \Tr_{A^n B^n}\! \left[ \Lambda_n\left(\Phi_2^{\otimes \floor{Rn}} \otimes \tau_n\right) \right] \qquad \forall\ n\, . \label{general_bounds_eq5}
\end{align}
We write
\bb
\sup_{\tau_{CD}} \left\{ M\left(\Phi_2^{\otimes \floor{Rn}} \otimes \tau_{CD} \right) - M(\tau_{CD}) \right\} &\geq M\left(\Phi_2^{\otimes \floor{Rn}} \otimes \tau_{C_n D_n} \right) - M(\tau_{C_n D_n}) \\
&\textgeq{(x)} M\left(\Lambda_n \left(\Phi_2^{\otimes \floor{Rn}} \otimes \tau_{C_n D_n} \right)\right) - M(\tau_{C_n D_n}) \\
&\textgeq{(xi)} M\left(\Tr_{C_n D_n}\Lambda_n \left(\Phi_2^{\otimes \floor{Rn}} \otimes \tau_{C_n D_n} \right) \right) \\
&\textgeq{(xii)} M\left(\rho^{\otimes n} \right) - n f(\e_n) \log d - g(\e_n) \, .
\label{general_bounds_eq6}
\ee
Here, (x)~is by $\FF$-monotonicity, (xi)~by strong super-additivity, and~(xii) by asymptotic continuity. Dividing by $n$ and taking the $\limsup$ as $n\to\infty$ yields
\bb
\limsup_{m\to\infty} \frac1m \sup_{\tau_{CD}} \left\{ M\left(\Phi_2^{\otimes m} \otimes \tau_{CD} \right) - M(\tau_{CD}) \right\} &\geq \limsup_{n\to\infty} \frac{1}{\floor{Rn}} \sup_{\tau_{CD}} \left\{ M\left(\Phi_2^{\otimes \floor{Rn}} \otimes \tau_{CD} \right) - M(\tau_{CD}) \right\} \\
&= \frac1R \limsup_{n\to\infty} \frac{1}{n} \sup_{\tau_{CD}} \left\{ M\left(\Phi_2^{\otimes \floor{Rn}} \otimes \tau_{CD} \right) - M(\tau_{CD}) \right\} \\
&\geq \frac1R \limsup_{n\to\infty} \frac{1}{n} \left( M\left(\rho^{\otimes n} \right) - n f(\e_n) \log d - g(\e_n) \right) \\
&= \frac{M^\infty(\rho)}{R}\, .
\label{general_bounds_eq7}
\ee
Taking the infimum over all achievable rates $R$ yields~\eqref{bound_cost_ssa}. Inequality~\eqref{bound_cost_mcc_ssa} is derived in a similar way.

If $M$ is only strongly super-additive and lower semi-continuous, we need to modify the inequality marked as~(xii) in~\eqref{general_bounds_eq6} as
\bb
\sup_{\tau_{CD}} \left\{ M\left(\Phi_2^{\otimes \floor{Rn}} \otimes \tau_{CD} \right) - M(\tau_{CD}) \right\}\ &\geq\ M\left(\Tr_{C_nD_n} \Lambda_n \left(\Phi_2^{\otimes \floor{Rn}} \otimes \tau_{C_n D_n} \right)\right) \\
&\textgeq{(xiii)}\ \sum_{j=1}^n M\left(\left(\Tr_{C_n D_n}\Lambda_n \left(\Phi_2^{\otimes \floor{Rn}} \otimes \tau_{C_n D_n} \right)\right)_j \right) \\
&\textgeq{(xiv)}\ n M\left(\left(\Tr_{C_n D_n}\Lambda_n \left(\Phi_2^{\otimes \floor{Rn}} \otimes \tau_{C_n D_n} \right)\right)_{j_n} \right) .
\label{general_bounds_eq8}
\ee
In~(xiii), as before, we denoted by $\left(\Tr_{C_n D_n}\Lambda_n \left(\Phi_2^{\otimes \floor{Rn}} \otimes \tau_{C_n D_n} \right)\right)_j$ the marginal of the overall state $\Tr_{C_n D_n}\Lambda_n \left(\Phi_2^{\otimes \floor{Rn}} \otimes \tau_{C_n D_n} \right)$ corresponding to the $j^\text{th}$ pair of systems $AB$ ($j=1,\ldots, n$), and in~(xiv) we picked $j_n\in\{1,\ldots, n\}$ such that
\bb
M\left(\left(\Tr_{C_n D_n}\Lambda_n \left(\Phi_2^{\otimes \floor{Rn}} \otimes \tau_{C_n D_n} \right)\right)_{j_n}\right) = \min_{j=1,\ldots, n} M\left(\left(\Tr_{C_n D_n}\Lambda_n \left(\Phi_2^{\otimes \floor{Rn}} \otimes \tau_{C_n D_n} \right)\right)_j\right) .
\ee
We can now divide both sides of~\eqref{general_bounds_eq8} by $n$ and take the $\limsup$ as $n\to\infty$, which yields as in~\eqref{general_bounds_eq7}
\bb
\limsup_{m\to\infty} \frac1m \sup_{\tau_{CD}} \left\{ M\left(\Phi_2^{\otimes m} \otimes \tau_{CD} \right) - M(\tau_{CD}) \right\} &\geq \frac1R \limsup_{n\to\infty} \frac{1}{n} \sup_{\tau_{CD}} \left\{ M\left(\Phi_2^{\otimes \floor{Rn}} \otimes \tau_{CD} \right) - M(\tau_{CD}) \right\} \\
&\geq \frac1R\, \limsup_{n\to\infty} M\left(\left(\Tr_{C_n D_n}\Lambda_n \left(\Phi_2^{\otimes \floor{Rn}} \otimes \tau_{C_n D_n} \right)\right)_{j_n} \right) \\
&\textgeq{(xv)} \frac{M(\rho)}{R}\, ,
\ee
where (xv)~is due to lower semi-continuity. This establishes~\eqref{bound_cost_ssa_unregularized}. The proof of~\eqref{bound_cost_mcc_ssa_unregularized} is entirely analogous.

Finally, if $M$ is strongly super-additive, completely additive, normalized, and lower semi-continuous, then on the one hand by~\eqref{bound_distillable_mcc_ssa}
\bb
E_{d,\,\FF^{mcc}}(\rho_{AB}) \leq 
\liminf_{n\to\infty} \frac1n \sup_{\tau_{CD} = \bigotimes_j \tau_{C_j D_j}} \left\{ M\left(\rho_{AB}^{\otimes n} \otimes \tau_{CD} \right) - \sum_j M(\tau_{C_jD_j}) \right\}\ \texteq{(xvi)}\ M(\rho_{AB})\, ,
\ee
where (xvi)~is due to additivity. On the other hand, by~\eqref{bound_cost_mcc_ssa_unregularized}
\bb
E_{c,\,\FF^{mcc}}(\rho_{AB})\ \geq\ 
\frac{M(\rho_{AB})}{\limsup_{m\to\infty} \frac1m \sup_{\tau_{CD} = \bigotimes_j \tau_{C_j D_j}} \left\{ M\left(\Phi_2^{\otimes m} \otimes \tau_{CD} \right) - \sum_j M(\tau_{C_j D_j}) \right\}}  \ \ \texteq{(xvii)}\ \ M(\rho_{AB}) \, ,
\ee
where (xvii)~is once again because of additivity and normalization. Putting all together thanks to~\eqref{hierarchy_distillable_and_cost} yields~\eqref{chain}.
\end{proof}

\subsection{Properties of the 
freely measured relative entropy of entanglement} \label{proofs_DKK_app}

This appendix contains proofs of known results concerning the freely measured relative entropy of entanglement, presented here for the sake of readability and completeness and with pointers to the original literature. The following has been established by Piani~\cite[Theorem~2]{piani_2009-1}.

\begin{manuallemma}{\ref{elementary_properties_DKK_lemma}}
For $\K\in \{\SEP, \PPT\}$ and $\mk\in \{\mlocc, \sep, \ppt\}$, the function $D_\K^\mk$ is invariant under local unitaries, convex, and monotonically non-increasing under LOCC. When $\K=\PPT$ and $\mk = \ppt$, it is also monotonically non-increasing under PPT operations.
\end{manuallemma}

\begin{proof}
The invariance under local unitaries follows straightforwardly from the fact that the set $\mk$ itself is invariant. As for convexity, in general
\bb
(\rho,\sigma)\mapsto D^\mk(\rho\|\sigma) \coloneqq \sup_{\MM\in \mk} D\left( \MM(\rho)\,\|\, \MM(\sigma)\right) 
\ee
is (jointly) convex, because it is a point-wise supremum of (jointly) convex functions. Now, $D_\K^\mk$ is obtained by taking the infimum of a jointly convex function over the second variable which runs over a convex set; hence, it must be convex itself. 

Monotonicity under LOCC holds because (a)~LOCC transform separable states into separable states, and (b)~the adjoint of an LOCC operation maps measurements in $\mk$ to measurements in $\mk$. Accordingly, monotonicity under PPT operations descends from the fact that (a')~PPT operations preserve PPT states, and (b')~the adjoint of a PPT operation maps PPT measurements to PPT measurements.

\end{proof}

The result below has been obtained independently in~\cite[Theorem~11]{Schindler2023}. The argument relies on a slightly modified version of the Alicki--Fannes--Winter trick~\cite{Alicki-Fannes, winter_2016-1}.

\begin{manuallemma}{\ref{asymptotic_continuity_DKK_lemma}}
Let $AB$ be a bipartite system of minimum local dimension $d\coloneqq \min\{|A|,|B|\}<\infty$. Let $\rho,\omega$ be two states on $AB$ at trace distance $\e\coloneqq \frac12 \left\|\rho-\omega\right\|_1$. For $\K\in \{\SEP, \PPT\}$ and any set of measurements $\mk$, it holds that
\begin{equation}
\left| D_\K^\mk(\rho) - D_\K^\mk(\omega)\right| \leq \epsilon \log d + g(\epsilon)\, ,
\tag{\ref{asymptotic_continuity_DKK}}
\end{equation}
where
\begin{equation}
g(x) \coloneqq (1+x)\log(1+x) - x\log x\, .
\tag{\ref{g}}
\end{equation}
\end{manuallemma}

\begin{proof}
We start by finding two states $\gamma_\pm$ on $AB$ such that
\bb
\rho-\omega \coloneqq \e (\gamma_+ - \gamma_-)\, .
\ee
Defining the state
\bb
\eta\coloneqq \frac{1}{1+\e}\, \rho + \frac{\e}{1+\e}\, \gamma_- = \frac{1}{1+\e}\, \omega + \frac{\e}{1+\e}\, \gamma_+\, ,
\ee
on the one hand using the convexity of $D_\K^\mk$ (Lemma~\ref{elementary_properties_DKK_lemma}) we see that
\bb
D_\K^\mk (\eta) \leq \frac{1}{1+\e}\, D_\K^\mk(\rho) + \frac{\e}{1+\e}\, D_\K^\mk(\gamma_-) \textleq{(i)} \frac{1}{1+\e}\, D_\K^\mk(\rho) + \frac{\e}{1+\e} \log d\, ,
\ee
where~(i) holds because $D_\K^\mk$ is always upper bounded by the relative entropy of entanglement, whose maximum value on states of a bipartite system $AB$ is of $\log \min\{|A|,|B|\}$~\cite{vedral_1997, vedral_1998}. On the other,
\bb
D_\K^\mk (\eta) &= \inf_{\sigma \in \K^1} \sup_{\MM\in \mk} D \left(\MM\left(\frac{1}{1+\e}\, \omega + \frac{\e}{1+\e}\, \gamma_+\right)\,\bigg\|\,\MM(\sigma)\right) \\
&\textgeq{(ii)} \inf_{\sigma \in \K^1} \sup_{\MM\in \mk} \left\{ \frac{1}{1+\e}\, D \left(\MM(\omega)\,\big\|\,\MM(\sigma)\right) + \frac{\e}{1+\e}\, D \left(\MM(\gamma_+)\,\big\|\,\MM(\sigma)\right) - h_2\left(\frac{\e}{1+\e}\right) \right\} \\
&\geq \inf_{\sigma \in \K^1} \sup_{\MM\in \mk} \left\{ \frac{1}{1+\e}\, D \left(\MM(\omega)\,\big\|\,\MM(\sigma)\right) - h_2\left(\frac{\e}{1+\e}\right) \right\} \\
&= \frac{1}{1+\e}\, D_\K^\mk(\omega) - h_2\left(\frac{\e}{1+\e}\right) .
\ee
Here, (ii)~is just the well-known fact that the von Neumann entropy is `not too concave'. Putting everything together we see that
\bb
\frac{1}{1+\e}\, D_\K^\mk(\omega) - h_2\left(\frac{\e}{1+\e}\right) \leq \frac{1}{1+\e}\, D_\K^\mk(\rho) + \frac{\e}{1+\e} \log d\, ,
\ee
which becomes
\bb
D_\K^\mk(\omega) - D_\K^\mk(\rho) \leq \e\log d+ g(\e)
\ee
after some elementary algebraic manipulations. By exchanging $\rho$ and $\omega$ we obtain also the opposite inequality, which together with the above yields~\eqref{asymptotic_continuity_DKK}, thus completing the proof. 
\end{proof}

The following result has been first obtained in~\cite[Proposition~4]{li_2014-1}, where it is even shown that~\eqref{max_ent_DKK} holds with equality.

\begin{manuallemma}{\ref{normalization_DKK_lemma}}
For $\K\in \{\SEP, \PPT\}$ and $\mk\in \{\mlocc,\sep,\ppt\}$, the maximally entangled state $\Phi_d$ on a $d\times d$ bipartite quantum system satisfies that
\begin{equation}
D_\K^\mk (\Phi_d) \geq \log (d+1) - 1\, .
\tag{\ref{max_ent_DKK}}
\end{equation}
\end{manuallemma}

\begin{proof}
Since $D_\K^\mk$ is convex and invariant under local unitaries, we can assume that the state in $\K$ achieving the minimum distance from $\Phi_d$ as measured by $D_\K^\mk$ is $\frac1d \Phi_d + \frac{d-1}{d}\, \frac{\id-\Phi_d}{d^2-1}$. Now, consider the measurement $(E, \id-E)$, where
\bb
E \coloneqq \sum_{i=1}^d \ketbra{ii}\, .
\ee
It can be implemented by LOCC, by simply measuring locally in the computational basis and accepting if and only if the two outcomes are the same. 

Now, introducing the function $D_2:[0,1]\times [0,1]\to \R_+\cup \{+\infty\}$ defined by
\bb
D_2(p\| q) \coloneqq p \log \frac{p}{q} + (1-p) \log \frac{1-p}{1-q}\, ,
\ee
we see that
\bb
D_\K^\mk (\Phi_d) &= D^\mk\left(\Phi_d\, \bigg\|\, \frac1d \Phi_d + \frac{d-1}{d}\, \frac{\id-\Phi_d}{d^2-1} \right) \\
&\geq D_2\left( \Tr\left[ E\Phi_d\right] \bigg\| \Tr\left[ E\left( \frac1d \Phi_d + \frac{d-1}{d}\, \frac{\id-\Phi_d}{d^2-1}\right) \right] \right) \\
&= D_2\left( 1\, \bigg\|\, \frac{2}{d+1} \right) \\
&= \log (d+1) - 1\, .
\ee
This concludes the proof.
\end{proof}

The result below is perhaps one of the deepest insights from the seminal work by Piani~\cite[Theorem~1]{piani_2009-1}. 

\begin{manuallemma}{\ref{Piani_inequality_lemma} {\normalfont \cite[Theorem~1]{piani_2009-1} (Piani's super-additivity-like inequality)}}
Let $\rho_{AA':BB'}$ be an arbitrary state over a pair of bipartite systems $AB$ and $A'B'$. For all pairs $(\K,\mk)$ satisfying~\eqref{pairs}, i.e., for all possible choices of $\K\in \{\SEP, \PPT\}$ and $\mk\in \{\mlocc,\sep,\ppt\}$ except 
for the pair $(\SEP, \ppt)$, it holds that
\begin{equation}
D_\K\left(\rho_{AA':BB'}\right) \geq D_\K\left(\rho_{A:B}\right) + D_\K^\mk\left(\rho_{A':B'}\right) ,
\tag{\ref{Piani_inequality}}
\end{equation}
where $D_\K$ is defined by~\eqref{DK}.
\end{manuallemma}

\begin{proof}
We start by observing that for all pairs of classical-quantum states $\rho_{XS} = \sum_x p_x \ketbra{x}\otimes \rho_S^x$ and $\sigma_{XS} = \sum_x q_x \ketbra{x}\otimes \sigma_S^x$, it holds that
\begin{align}
D\left(\rho_{XS}\| \sigma_{XS}\right) &= \sum_x p_x \Tr \left[\rho_S^x \left( \log(p_x \rho_S^x) - \log (q_x\sigma_S^x) \right) \right] \nonumber \\
&= \sum_x p_x \log\frac{p_x}{q_x} + \sum_x p_x \Tr \left[\rho_S^x \left( \log\rho_S^x - \log \sigma_S^x \right) \right] \nonumber \\
&= D(\rho_X \| \sigma_X) + \sum_x p_x\, D\big(\rho_S^x\, \big\|\, \sigma_S^x\big) \label{superadditivity_semiclassical_relent} \\
&\geq D(\rho_X \| \sigma_X) + D\left( \sumno_x p_x \rho_S^x\, \Big\| \sumno_x p_x \sigma_S^x \right) \nonumber \\
&= D(\rho_X \| \sigma_X) + D\left( \rho_S\, \Big\| \sumno_x p_x \sigma_S^x\right) . \nonumber
\end{align}
Before we proceed, let us observe that for all $\sigma_{AA':BB'} \in \K$ and $\MM'_{A':B'}(\cdot) = \sum_x \Tr \left[ F_x^{A':B'} (\cdot) \right] \ketbra{x} \in \mk$, with $(\K,\mk)$ satisfying~\eqref{pairs}, it holds that
\bb
\Tr_{A':B'}\left[ F_x^{A':B'} \sigma_{AA':BB'} \right] \in \K_{A:B}\qquad \forall\ x\, .
\label{compatibility}
\ee
This key feature is called the \emph{compatibility condition} in~\cite{brandao_2020}. Putting all together,
\bb
D_\K (\rho_{AA':BB'}) &= \inf_{\sigma_{AA':BB'}\in \K^1} D\left( \rho_{AA':BB'} \,\|\, \sigma_{AA':BB'}\right) \\
&\textgeq{(i)} \inf_{\sigma_{AA':BB'}\in \K^1} \sup_{\MM'_{A':B'}\in \mk} D\big( \MM'_{A':B'}(\rho_{AA':BB'})\, \big\|\, \MM'_{A':B'}(\sigma_{AA':BB'}) \big) \\
&\textgeq{(ii)} \inf_{\sigma_{AA':BB'}\in \K^1} \sup_{\MM'_{A':B'}\in \mk} \Big\{ D\left( \rho_{A:B}\, \Big\|\, \sumno_x \Tr\left[ F_x^{A':B'} \rho_{A':B'} \right] \widetilde{\sigma}^x_{A:B} \right) \\
&\hspace{22.5ex} + D\big( \MM'_{A':B'}(\rho_{A':B'})\, \big\|\, \MM'_{A':B'}(\sigma_{A':B'}) \big) \Big\} \\
&\textgeq{(iii)} \inf_{\sigma_{A':B'}\in \K^1} \sup_{\MM'_{A':B'}\in \mk} \Big\{ D_\K(\rho_{A:B}) + D\big( \MM'_{A':B'}(\rho_{A':B'})\, \big\|\, \MM'_{A':B'}(\sigma_{A':B'}) \big) \Big\} \\
&= D_\K (\rho_{A:B}) + D_\K^\mk (\rho_{A':B'})\, .
\ee
Here, in~(i) we introduced a measurement channel $\MM'_{A':B'}(\cdot) = \sum_x \Tr \left[ F_x^{A':B'} (\cdot) \right] \ketbra{x} \in \mk$ and applied the data processing inequality for the relative entropy, in~(ii) we made use of~\eqref{superadditivity_semiclassical_relent}, introducing also the post-measurement state
\bb
\widetilde{\sigma}_{A:B}^x \coloneqq \frac{1}{\Tr \left[ F_x^{A':B'} \sigma_{A':B'} \right]}\, \Tr_{A'B'}\left[ F_x^{A':B'} \sigma_{AA':BB'} \right] ,
\ee
and finally (iii)~comes from the compatibility condition~\eqref{compatibility} due to the convexity of $\K$.
\end{proof}

\end{document}